\documentclass[english,journal]{IEEEtran}

\usepackage{xcolor}
\usepackage{bm}
\usepackage{stfloats}
\usepackage{mathtools}
\usepackage{enumitem}
\usepackage{siunitx}

\usepackage[style=ieee, doi=false, isbn=false]{biblatex}
\usepackage{graphicx}
\usepackage{url}      
\usepackage{amsmath}  
\usepackage{amsfonts,amssymb}
\usepackage{microtype}
\usepackage{babel}
\usepackage{csquotes}
\usepackage{comment}

\newtheorem{lemma}{{Lemma}}

\newtheorem{proposition}{{Proposition}}

\newenvironment{proof}{{\noindent\it Proof:}}{\hfill $\square$\par}
\newenvironment{proofp3}{{\noindent\it Proof of Proposition 3:}}{\hfill $\square$\par}

\definecolor{purple}{RGB}{128,0,128}

\let\hat\widehat
\let\tilde\widetilde

\begin{document}

\title{Joint Localization and Orientation Estimation in Millimeter-Wave MIMO OFDM Systems via \\ Atomic Norm Minimization}

\author{Jianxiu~Li,~\IEEEmembership{Graduate Student Member,~IEEE,}
        Maxime~Ferreira~Da~Costa,~\IEEEmembership{Member,~IEEE,}
        and~Urbashi~Mitra,~\IEEEmembership{Fellow,~IEEE}
\thanks{J. Li, M. Ferreira Da Costa, and U. Mitra are with the Department of Electrical and Computer Engineering, University of Southern California, Los Angeles, CA 90089, USA (e-mails: {jianxiul, mferreira, ubli}@usc.edu).}
\thanks{{This paper will be presented in part at the 2022 IEEE International Conference on Acoustics, Speech and Signal Processing (ICASSP 2022)~\cite{Liicassp}}. This work has been funded in part by one or more of the following: Cisco Foundation 1980393, ONR N00014-15-1-2550, ONR 503400-78050, NSF CCF-1410009, NSF CCF-1817200, NSF CCF-2008927, Swedish Research Council 2018-04359, ARO W911NF1910269, and DOE DE-SC0021417.}
}

\markboth{}
{LI, FERREIRA DA COSTA, AND MITRA: ATOMIC NORM BASED LOCALIZATION AND ORIENTATION ESTIMATION FOR MILLIMETER-WAVE MIMO OFDM SYSTEMS}

\maketitle

\begin{abstract}
Herein, an atomic norm based method for accurately estimating the location and orientation of a target from millimeter-wave multi-input-multi-output (MIMO) orthogonal frequency-division multiplexing (OFDM) signals is presented. A novel virtual channel matrix is introduced and an algorithm to extract localization-relevant channel parameters from its atomic norm decomposition is designed. Then, based on the extended invariance principle, a weighted least squares problem is proposed to accurately recover the location and orientation using both line-of-sight and non-line-of-sight channel information. The conditions for the optimality and uniqueness of the estimate and theoretical guarantees for the estimation error are characterized for the noiseless and the noisy scenarios. Theoretical results are confirmed via simulation.
 Numerical results investigate the robustness of the proposed algorithm to incorrect model order selection or synchronization error, and highlight performance improvements over a prior method. The resultant performance nearly achieves the Cram\'{e}r-Rao lower bound on the estimation error.
\end{abstract}


\IEEEpeerreviewmaketitle

\section{Introduction}
\label{sec:intro}
\IEEEPARstart{M}{illimeter}-wave (mmWave) communications are of strong interest for fifth generation multi-input-multi-output (MIMO) systems due to the significant bandwidth afforded by small wavelengths~\cite{Rappaport}. One challenge with these systems is path loss; however, this effect 
can be mitigated by beamforming~\cite{Hur}.  In addition to large bandwidth, mmWave MIMO also experiences limited multipath, thus making this signaling a natural candidate for localization~\cite{lemic2016localization,Hua,Zhou,Saloranta,Shahmansoori,Zhu,Fan,Zheng}. In~\cite{Saloranta,Shahmansoori},  localization based on classic compressed sensing is pursued by exploiting multipath sparsity.
However, the performance is limited by quantization error and grid resolution.  In~\cite{Shahmansoori}, a space-alternating generalized expectation maximization algorithm is proposed to refine the channel estimates for localization, initialized with the channel parameters that are coarsely estimated via a modified distributed compressed sensing simultaneous orthogonal matching pursuit (DCS-SOMP) scheme~\cite{Duarte}. However, this method suffers from local minima when the signal-to-noise ratio (SNR) is low or when the initialization is not sufficiently accurate. {To address the grid mismatch, sparse Bayesian learning based methods \cite{Zhu,Fan} are designed for mmWave MIMO localization, relying on the knowledge of the prior distribution of the channel coefficients. For a more realistic, band-limited, frequency selective channel model, a weighted orthogonal matching pursuit strategy is proposed in \cite{Zheng}, which operates in parallel in the angular and delay domains to improve the localization accuracy with lower requirement on the grid resolution.}

{In contrast, atomic norm minimization (ANM)~\cite{chi2020harnessing}, also known as total variation minimization, has emerged as a useful convex optimization framework for estimating continuous valued parameters without relying on discretization. 
Of late, this method has been actively studied in the context of recovering the parameters of a sparse sum of complex exponentials~\cite{candes2014towards}, as we do herein. 
ANM is known to achieve near-optimal denoising rates in the signal space in that context~\cite{tang_near_2015}, and to recover the correct model-order under a separation condition on the parameters~\cite{duval2020characterization,da2020stable}. Error-bounds on the parameter estimates have been derived under a Gaussian 
noise assumption~\cite{li_approximate_2018}.} 
Furthermore, this framework has also been successfully generalized to recover multi-dimensional signals \cite{chi2014compressive,Tsai} {as we do herein.} In \cite{chi2014compressive}, the two-dimensional case is studied and a higher dimensional case is discussed in \cite{Tsai}.  We observe that our semi-definite program formulation is different from that suggested in \cite{Tsai}.

The ANM framework has previously been employed for the purpose of localization~\cite{Heckel, Wu1,Wu2,TangW}, but in the absence of multipath. In~\cite{Tsai,Beygi,Elnakeeb}, ANM is used for channel estimation, but can not be used for localization as the time-of-arrival (TOA) is not considered in the model.
In contrast to ~\cite{Beygi,Elnakeeb,Li} that focus on the estimation of one-dimensional signals, this paper considers high-accuracy estimation of the localization and orientation from signals using mmWave MIMO orthogonal frequency-division multiplexing (OFDM) signaling. A {\em multi-dimensional} ANM based estimator is proposed to simultaneously recover  all the relevant parameters by harnessing the structure of a novel virtual channel matrix. {We underscore that our methods herein do not require the prior, Bayesian model information for \cite{Zhu,Fan} and we have considered other atomic norm based schemes \cite{Beygi,Elnakeeb,Li} which can exploit the effect of practical pulse shapes in \cite{Zheng}}.

The main contributions of this paper are:
\begin{enumerate}[wide=0pt]
\item A novel {\it virtual channel matrix} capturing the geometry of the paths is designed for mmWave MIMO OFDM channels. 
\item Based on the sparse structure of the virtual channel matrix, a multi-dimensional atomic norm channel estimator is proposed to simultaneously estimate the TOAs, angle-of-arrivals (AOA), and angle-of-departures (AOD) with super-resolution, \emph{i.e.} without relying on a discretization of the search space.
\item Sufficient conditions for exact recovery of the parameters are given in Propositions~\ref{prop:uniqueandoptimal} and \ref{prop:exactRecovery} for the noiseless case; in the presence of noise, an upper bound is derived for the estimation errors in Proposition \ref{errorboundG}, suggesting the appropriate value of a key regularization parameter to improve performance. 
\item Location and orientation are accurately inferred from the solution of ANM through a weighted non-linear least squares scheme,   where the designed weight matrix is compatible with the ANM channel estimator.
\item The theoretical analyses for the proposed scheme are validated via simulation, and the comparisons with the DCS-SOMP method~\cite{Shahmansoori} are presented. It is shown that the proposed scheme offers more than $7$dB gain with respect to the root-mean-square error (RMSE) of estimation when a small number of antennas are employed;  furthermore,  the proposed method nearly achieves the 
 Cram\'{e}r-Rao lower bound (CRLB)~\cite{Shahmansoori} in the studied cases.
\item The effects of incorrect model order selection and synchronization error on the estimation accuracy are numerically investigated to show the robustness of the proposed scheme.
\end{enumerate}

{The present work completes our previous work~\cite{Liicassp} with an analysis of the optimality and uniqueness of the estimate in the noiseless case. Furthermore, we derive an expression for the regularization parameter under a white Gaussian noise assumption as a function of the pilot signals to achieve near-optimal denoising rates, with extended derivations and proofs.
The robustness of the proposed method to incorrect model order selection and  synchronization error is also studied via simulation; the proposed approach is shown to be durable to such mismatch.}

The rest of this paper is organized as follows. Section~\ref{sec:signal} introduces the signal model for the mmWave MIMO OFDM systems. Section~\ref{sec:structure} presents the design of a novel virtual channel matrix, which captures geometric parameters of the signal model. In Section~\ref{sec:method}, the structure of the proposed virtual channel matrix is explicitly exploited for the design of a multi-dimensional atomic norm based channel estimator. Sufficient conditions\ for exact recovery in absence of noise is investigated in Section~\ref{sec:channelestimation}, and a denoising bound  is derived in Section~\ref{sec:errorRate} under a white Gaussian noise assumption. Section \ref{sec:semidefiniteApproximation} presents a semidefinite representation of ANM that can be solved using off-the-shelf convex solvers. Section \ref{sec:channelParametersEstimation} and \ref{sec:EXIP} propose a high accuracy method to recover the location and orientation of the target from the solution of the ANM program based on the Vandermonde factorization of Toeplitz matrices~\cite{Yang2} and the extended invariance principle (EXIP)~\cite{Stoica}. Numerical results are given in Section~\ref{sec:sim} to verify the theoretical analyses and highlight the performance improvements with respect to the previous method. Finally, a conclusion is drawn in Section~\ref{sec:con}. Appendices A and B provide the proofs of two propositions of this work.

Scalars are denoted by lower-case letters, $x$ and column vectors by bold letters $\bm{x}$. The $i$th element of $\bm{x}$ is denoted by $\bm{x}[i]$. Matrices are denoted by bold capital letters, $\bm X$. $\bm{x}_l$ denotes the $l$th column of $\bm{X}$ and $\boldsymbol{X}[i, j]$ is the ($i$, $j$)th element. The operators $\lfloor{x}\rfloor$, $|x|$, $\|\bm  x\|_{1}$,$\|\bm  x\|_{2}$, $\operatorname{sign}(x)$, $\operatorname{Re}(x)$, $\operatorname{diag}(\mathcal{A})$ represent the largest integer that is less than $x$, the magnitude of $x$, the $\ell_1$ norm of $\bm x$, the $\ell_2$ norm of $\bm x$, the complex sign of $x$ defined as $\operatorname{sign}(x) = \frac{x}{\left\Vert x \right\Vert_2}$, the real part of $x$, a diagonal matrix whose diagonal elements are given by $\mathcal{A}$. ${\|\boldsymbol{X}\|}_{F}$ is the Frobenius norm of $\boldsymbol{X}$. $\langle\boldsymbol{A},\boldsymbol{X}\rangle$ stands for the usual matrix inner product  while $\langle\boldsymbol{A},\boldsymbol{X}\rangle_{R}\triangleq\operatorname{Re}(\langle\boldsymbol{A},\boldsymbol{X}\rangle)$. $\boldsymbol{I}_{d}$ and $\boldsymbol{0}_{d}$ stand for a $d\times d$ identity matrix and a $d\times d$ zero matrix, respectively. The notations  $\otimes$ and $\mathbb{E}\{\cdot\}$ denote the Kronecker product and the expectation of a random variable. The operators $\operatorname{rank}(\cdot), \operatorname{Tr}(\cdot), (\cdot)^\mathrm{T}$, and $(\cdot)^{\mathrm{H}}$ are defined as the rank of a matrix, the trace of a matrix, the transpose of a matrix or vector, and the conjugate transpose of a vector or matrix, respectively.


\section{mmWave MIMO OFDM Narrowband Channels}

\subsection{Signal Model}\label{sec:signal}
We adopt the narrowband channel model of~\cite{Shahmansoori}, where a single base station (BS) is equipped with $N_t$ antennas and a target has $N_r$ antennas. The locations of the BS and the target are denoted by $\boldsymbol{q}=\left[q_{x}, q_{y}\right]^{\mathrm{T}} \in \mathbb{R}^{2}$ and $\boldsymbol{p}=\left[p_{x}, p_{y}\right]^{\mathrm{T}} \in \mathbb{R}^{2}$, respectively, where $\boldsymbol{q}$ is known while $\boldsymbol{p}$ is to be estimated. In addition, there is an unknown orientation of the target's antenna array, denoted by $\theta_o$. Assume that one line-of-sight (LOS) path and $K$ non-line-of-sight (NLOS) paths exist in the mmWave MIMO OFDM channel. The $k$-th NLOS path is produced by a scatterer at an unknown location $\boldsymbol{s}_k=\left[s_{k,x}, s_{k,y}\right]^{\mathrm{T}} \in \mathbb{R}^{2}$. 
\begin{figure}[t]
    \centering
    \includegraphics[width=0.95\columnwidth]{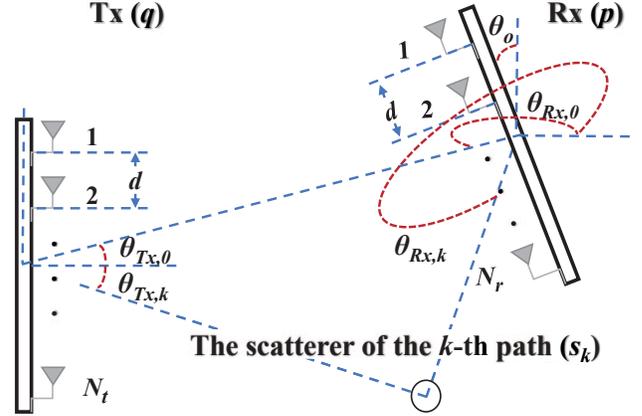} 
    \caption{~{System model.}}
    \label{systemmodel}
\end{figure}
Denoting by $N_s$ the number of the sub-carriers, we transmit $G$ OFDM pilot signals with carrier frequency $f_c$ and bandwidth $B\ll f_c$. Given  the $g$-th pilot signal over the $n$-th sub-carrier $\boldsymbol{x}^{(g,n)}$ \footnote{The pilot signals $\boldsymbol{x}^{(g,n)}\in\mathbb{C}^{N_t}$ is a general expression that permits the incorporation of the beamforming matrix, the design of which is beyond the scope of this paper.}, the $g$-th received signal over the $n$-th sub-carrier is given by
\begin{equation}
\boldsymbol{y}^{(g,n)}=\boldsymbol{H}^{(n)}  \boldsymbol{x}^{(g,n)}+\boldsymbol{w}^{(g,n)},\label{rsignal}
\end{equation}
where $\boldsymbol{w}^{(g,n)} \sim \mathcal{CN}(\bm{0},\sigma^2 \bm{I}_{N_r})$ is an independent, zero-mean, complex Gaussian vector with variance $\sigma^2$. We denote by $\bm{a}_N(\delta)$ the Fourier vector, {\it i.e.,}
\(
    \bm{a}_N(\delta) \triangleq \frac{1}{\sqrt{N}} \left[1, e^{-j 2\pi \delta}, \dots, e^{-j 2\pi (N-1)\delta}\right]^{\mathrm{T}}.
\)
The $n$-th sub-carrier channel matrix $\boldsymbol{H}^{(n)}$ with $0\leq n\leq N_s-1$ is then given by
\begin{equation}
\boldsymbol{H}^{(n)}\triangleq\sum_{k=0}^{K}\gamma_k e^{\frac{-j 2\pi n\tau_k}{N_s T_{s}}}\boldsymbol{\alpha}(\theta_{\mathrm{Rx},k})\boldsymbol{ \beta}\left(\theta_{\mathrm{Tx},k}\right)^{\mathrm{H}},
\label{channelmatrix_subcarrier}
\end{equation}
where $T_s\triangleq\frac{1}{B}$ is the sampling period, $\gamma_k\triangleq\sqrt{N_tN_r}\frac{h_{k}}{\sqrt{\rho_{k}}}$ is the channel coefficient of the $k$-th path, while $\rho_k$ and $h_k$ represents the path loss and the complex channel gain, respectively. The TOA of the $k$th path is denoted as $\tau_k$. The index $k=0$ represents the LOS path. It is assumed that $\tau_0,\tau_1,\cdots,\tau_K<NT_s$. The steering vectors of the system, {\it i.e.,} $\boldsymbol{\alpha}(\theta_{\mathrm{Rx}})$ and $\boldsymbol{\beta}(\theta_{\mathrm{Tx}})$, are defined as $\boldsymbol{\alpha}(\theta_{\mathrm{Rx}})\triangleq \bm{a}_{N_r}\left(\frac{d \sin(\theta_{\mathrm{Rx}})}{\lambda_c} \right)$, $\boldsymbol{\beta}(\theta_{\mathrm{Tx}}) \triangleq \bm{a}_{N_t}\left(\frac{d \sin(\theta_{\mathrm{Tx}})}{\lambda_c} \right)$, where $d$ is the distance between antennas and $\lambda_c\triangleq \frac{c}{f_c}$ is the wavelength with $c$ being the speed of light. From the geometry shown in Fig. \ref{systemmodel}, the TOA, AOA, and AOD of each path, {\it i.e.,} $\tau_{k}, \theta_{\mathrm{Rx}, k}, \text{and } \theta_{\mathrm{Tx}, k}$, with $0\leq k\leq K$, are 
\begin{subequations}\label{eq:geometricMapping}
\begin{align}
\tau_{0} &=\frac{\|\boldsymbol{p}-\boldsymbol{q}\|_{2}} { c},\label{tau0}\\
\tau_{k} &=\frac{\left\|\boldsymbol{q}-\boldsymbol{s}_{k}\right\|_{2} +\left\|\boldsymbol{p}-\boldsymbol{s}_{k}\right\|_{2}} {c}, \quad k>0,\label{tauk} \\
\theta_{\mathrm{Tx}, 0} &=\arctan \left(\frac{p_{y}-q_{y} }{p_{x}-q_{x}}\right),\label{thetatx0}\\ 
\theta_{\mathrm{Tx}, k} &=\arctan \left(\frac{s_{k,y}-q_{y}} {s_{k,x}-q_{x}}\right),  \quad k>0,\label{thetatxk}\\
\theta_{\mathrm{Rx}, 0} &=\pi+\arctan \left(\frac{p_{y}-q_{y}} {p_{x}-q_{x}}\right)-\theta_o,\label{thetarx0}\\
\theta_{\mathrm{Rx}, k} &=\pi+\arctan \left(\frac{p_{y}-s_{k,y}}{p_{x}-s_{k,x}}\right)-\theta_o, \quad k>0.
\label{thetarxk}
\end{align}
\end{subequations}
Stacking the received signal of Equation~(\ref{rsignal}) in a matrix form leads to
\begin{equation}
        \boldsymbol{Y} = \boldsymbol{H}\boldsymbol{X} + \boldsymbol{W},\label{signalmodelstacked} 
\end{equation}
where $\boldsymbol{H}\triangleq\operatorname{diag}\left(\left\{{\boldsymbol{H}}^{(0)}, {\boldsymbol{H}}^{(1)},\cdots, {\boldsymbol{H}}^{(N_s-1)}\right\}\right)$ and
\begin{align*}
    \boldsymbol{Y} \triangleq{}&\left[\left(\boldsymbol{Y}^{(0)}\right)^{\mathrm{T}}, \left(\boldsymbol{Y}^{(1)}\right)^{\mathrm{T}}, \ldots, \left(\boldsymbol{Y}^{(N_s-1)}\right)^{\mathrm{T}}\right]^{\mathrm{T}},\\
    \boldsymbol{X} \triangleq{}& \left[\left(\boldsymbol{X}^{(0)}\right)^{\mathrm{T}},\left(\boldsymbol{X}^{(1)}\right)^{\mathrm{T}}, \ldots, \left(\boldsymbol{X}^{(N_s-1)}\right)^{\mathrm{T}}\right]^{\mathrm{T}},\\
    \boldsymbol{W} \triangleq{}&\left[\left(\boldsymbol{W}^{(0)}\right)^{\mathrm{T}}, \left(\boldsymbol{W}^{(1)}\right)^{\mathrm{T}}, \ldots, \left(\boldsymbol{W}^{(N_s-1)}\right)^{\mathrm{T}}\right]^{\mathrm{T}},
\end{align*}
with
$\boldsymbol{Y}^{(n)} \triangleq \left[\boldsymbol{y}^{(1,n)}, \boldsymbol{y}^{(2,n)}, \ldots, \boldsymbol{y}^{(G,n)}\right]$, $\boldsymbol{X}^{(n)} \triangleq\left[ \boldsymbol{x}^{(1,n)},  \boldsymbol{x}^{(2,n)}, \right.$ 
$\left.\ldots,  \boldsymbol{x}^{(G,n)}\right]$, 
and $\boldsymbol{W}^{(n)} \triangleq\left[ \boldsymbol{w}^{(1,n)},  \boldsymbol{w}^{(2,n)}, \ldots,  \boldsymbol{w}^{(G,n)}\right]$.
Furthermore, it is assumed that the receiver knows the transmitted symbols $\boldsymbol{X}$. The localization and orientation estimation problem is defined as recovering with optimal precision the unknown target position $\bm{p}$ and orientation angle $\theta_o$ from the knowledge of the observation $\bm{Y}$ and the pilot sequence $\bm{X}$. 


\subsection{Structure of the Virtual Channel Matrix}\label{sec:structure}
In this subsection, we present a novel construction of a virtual channel matrix $\bm{H}_v$ which captures the low-rank properties of mmWave MIMO OFDM narrowband channel and jointly incorporates the TOA, AOA and AOD of the problem geometry on each of the sub-carriers. In the sequel, $N_s$ is assumed to be an odd integer for {notational} convenience, and we define the block-matrices $\bm{Q}_{i,m}$, for $m=1,2$, with the $(j_m,z_m)$-th block, {\it i.e.,} $\bm{Q}_{i,m}^{(j_m,z_m)}$, given by 
    \begin{subequations}\label{Qfunctions}
    \begin{align}
    \boldsymbol{Q}^{(j_1,z_1)}_{i,1}=\left\{
    \begin{aligned}
    \boldsymbol{I}_{N_r},&\ \text{  if } j_1=i \text{ and } z_1 =   \left\lfloor\frac{i+1}{2}\right\rfloor,\\\
    \boldsymbol{0}_{N_r},&\ \text{ otherwise, }\\
    \end{aligned}
    \right.\\
    \boldsymbol{Q}^{(j_2,z_2)}_{i,2}=\left\{
    \begin{aligned}
    \boldsymbol{I}_{N_t},&\ \text{  if } j_2=1+\left\lfloor\frac{i}{2}\right\rfloor \text{ and } z_2 = i, \\
    \boldsymbol{0}_{N_t},&\ \text{ otherwise, }\\
    \end{aligned}
    \right.
    \end{align}
    \end{subequations}
    where $1\leq i,j_1, z_2 \leq N_s$ and $1\leq z_1, j_2 \leq \frac{N_s+1}{2}$.
The virtual channel matrix $\bm{H}_v$ is given by
\begin{equation}
\begin{aligned}
&\boldsymbol{H}_v\triangleq\sum_{k=0}^{K}l_k \left({\boldsymbol{\xi}}(\tau_k)\otimes\boldsymbol{\alpha}(\theta_{\mathrm{Rx},k})\right)\left({\boldsymbol{\xi}}(-\tau_k)\otimes\boldsymbol{ \beta}\left(\theta_{\mathrm{Tx},k}\right)\right)^{\mathrm{H}},
\end{aligned}\label{virtualchannel}
\end{equation}
where \mbox{$l_k \triangleq \frac{(N_s+1)\gamma_k}{2}$} and 
${\boldsymbol{\xi}}(\tau)\triangleq \bm{a}_{\frac{N_s+1}{2}}(\frac{\tau}{N_s T_s})$.
Next, we list key properties of the matrix $\bm{H}_v$ defined in Equation~\eqref{virtualchannel}:
\begin{itemize}
    \item[O1)] ${\boldsymbol{H}}_v$ has rank $K+1$ provided that that $K+1 \leq \min(N_r, N_t)< \min\left(\frac{(N_s+1)N_r}{2}, \frac{(N_s+1)N_r}{2}\right)$;
    \item[O2)] ${\boldsymbol{H}}_v$ has the same rank as $\bm{H}^{(n)}$, for any $n$; 
    \item[O3)] ${\boldsymbol{H}}_v$ is a block Hankel matrix, {\it i.e.}, the ($i,j$)-th $N_r \times N_t$ block matrix ${\boldsymbol{H}}_v^{(i,j)}$ of ${\boldsymbol{H}}_v$ satisfies ${\boldsymbol{H}}_v^{(i,j)}={\boldsymbol{H}}_v^{(k,z)}$ if $i+j=k+z$ for any $1\leq i,j,k,z\leq \frac{N_s+1}{2} $;
    \item[O4)] ${\boldsymbol{H}}_v^{(i,j)}={\boldsymbol{H}}^{(i+j-2)}$ holds for any $1\leq i,j\leq \frac{N_s+1}{2} $, defining an automorphism $g$ between ${\boldsymbol{H}}_v$ and $\boldsymbol{H}$, {\it i.e.,} 
    \begin{equation}
    \boldsymbol{H}={g}(\boldsymbol{H}_v) \triangleq \sum_{i=1}^{N_s}\boldsymbol{Q}_{i,1} \boldsymbol{H}_v \boldsymbol{Q}_{i,2}.
    \end{equation}
\end{itemize}
Note that, there may exist other possible constructions of a virtual channel matrix capturing the low-rank property of the MIMO OFDM channel. However, our construction is compact, and  symmetrically captures the receiver and the transmitter parameters. This representation will be shown to yield near-optimal reconstruction in simulation studies (see Section~\ref{sec:sim}).

\section{Atomic Norm Minimization Based {Channel, Location, and Orientation} Estimation}\label{sec:method} 


In this section, we first propose a tractable optimization method to estimate the virtual channel matrix as well as the individual channel parameters based on ANM. The optimality and uniqueness of the channel estimates are established for the noiseless case in Propositions~\ref{prop:uniqueandoptimal} and \ref{prop:exactRecovery}. Furthermore, the expected error rate (EER) on the parameters of the atomic norm denoiser is derived under a white Gaussian noise assumption in Proposition \ref{errorboundG}. Finally, given the mappings in Equations~(\ref{tau0})-(\ref{thetarxk}) between $\boldsymbol{{\eta}}\triangleq \left\{ \{\tau_0,\theta_{\mathrm{Tx},0},\theta_{\mathrm{Rx},0}\},\cdots, \{\tau_K,\theta_{\mathrm{Tx},K},\theta_{\mathrm{Rx},K}\}\right\}$ and $\tilde{\boldsymbol{\eta}}\triangleq\left\{\boldsymbol q, \theta_o, \boldsymbol{s}_1,\boldsymbol{s}_2,\cdots, \boldsymbol{s}_K\right\}$, we estimate the orientation and location of the target via a non-linear weighted least squares problem based on the EXIP~\cite{Stoica}, which is compatible with our proposed ANM based channel estimator. The proposed strategy is denoted as  \textit{Location-Orientation-Channel estimation via Multi-dimensional Atomic Norm} (LOCMAN).

\subsection{Exact Recovery in Noiseless Settings}\label{sec:channelestimation}

The virtual channel matrix $\bm{H}_v$ in Equation~\eqref{virtualchannel} is a sparse summation of rank-one matrices lying in the atomic set $\bm{\mathcal{A}}$ defined by 
\begin{equation}
\begin{aligned}\label{eq:atomicSet}
&\bigg\{ \boldsymbol{A}\left(\tau,\theta_{\mathrm{Rx}}, \theta_{\mathrm{Tx}},\phi\right) \triangleq e^{j\phi}\boldsymbol{\chi}(\tau,\theta_{\mathrm{Rx}})\boldsymbol{\zeta}(\tau,\theta_{\mathrm{Tx}})^{\mathrm{H}}\mid \phi\in[0,2\pi),\\
&\quad\frac{d \sin \left(\theta_{\mathrm{Rx}}\right)}{\lambda_{c}}, {\frac{d \sin \left(\theta_{\mathrm{Tx}}\right)}{\lambda_{c}} }\in(-\frac{1}{2},\frac{1}{2}],\frac{\tau}{N_sT_s } \in(0,1]\bigg\},
\end{aligned}
\end{equation}
where ${{\boldsymbol{\chi}}}(\tau,\theta_{\mathrm{}})\triangleq{\boldsymbol{\xi}}(\tau)\otimes\boldsymbol{ \alpha}\left(\theta_{\mathrm{}}\right)
$ and ${{\boldsymbol{\zeta}}}(\tau,\theta_{\mathrm{}})\triangleq{\boldsymbol{\xi}}(-\tau)\otimes\boldsymbol{ \beta}\left(\theta_{\mathrm{}}\right).$ 
The \emph{atomic norm} of $\bm{H}_v$, denoted as $\left\Vert \bm{H}_v \right\Vert_{\mathcal{A}}$, is defined as the Minkowski functional associated {with} $\bm{\mathcal{A}}$ and is given by
\begin{align}
\left\Vert \bm{H}_v \right\Vert_{\mathcal{A}} &= \inf_{t>0} \{ \bm{H}_v \in t \operatorname{conv}(\bm{\mathcal{A}}) \} \nonumber \\
&\hspace{-.3in}=\inf\left\{\sum_{k}\left|{\tilde{l}_{k}}\right| \mid \boldsymbol{H}_v=\sum_{k} \left|\tilde{l}_{k}\right| \boldsymbol{A}\left(\tau_k,\theta_{\mathrm{Rx},k}, \theta_{\mathrm{Tx},k},\phi_k\right)\right\}\label{defatomicnorm}
\end{align}
The theory of atomic norm minimization \cite{chi2020harnessing} enables the estimation of parameters in the sparse decomposition of $\bm{H}_v$ via the \emph{atomic decomposition} $ \{\tau_k,\theta_{\mathrm{Rx},k}, \theta_{\mathrm{Tx},k},\phi_k\}$ realizing the infimum of Equation~\eqref{defatomicnorm}.

In absence of noise, atomic norm minimization consists of  estimating the virtual channel matrix by the matrix with minimal atomic norm which is consistent with the observation model in Equation~\eqref{signalmodelstacked}, yielding the convex program
\begin{equation}
\begin{aligned}
\boldsymbol{\hat{{H}}}_v = \mathop{\arg\min}_{\boldsymbol{H}_v}\quad & {\|\boldsymbol{H}_v\|}_{\mathcal{A}} \\
\text {s.t.} \quad & \boldsymbol{Y}=g(\boldsymbol{H}_v)\boldsymbol{X}.
\end{aligned}\label{optprobnoiseless}
\end{equation}
Of crucial interest is to understand when exact recovery is achieved \emph{i.e.} when $\hat{\bm{H}}_v = \bm{H}_v$. This property is well-understood to be related to the existence of a dual solution {satisfying} certain interpolation properties ( \textit{i.e.} the dual certificate).
The following proposition provides the conditions to ensure the optimality and uniqueness of the solution of Equation~(\ref{optprobnoiseless}).
\begin{proposition}[Dual certificate]\label{prop:uniqueandoptimal} Suppose that there exists a matrix $\bm{\Lambda}$ such that the function
\begin{equation}
\Pi(\tau,\theta_\mathrm{Rx},\theta_\mathrm{Tx})\triangleq \left\langle\sum_{i=1}^{{N_s}{}}\boldsymbol{Q}_{i,1}^\mathrm{H} {\boldsymbol{\Lambda}}{\boldsymbol{X}}^\mathrm{H} \boldsymbol{Q}_{i,2}^\mathrm{H},\boldsymbol{A}(\tau,\theta_\mathrm{Rx},\theta_\mathrm{Tx},0)\right\rangle_R,
\label{dualpoly}
\end{equation}
satisfies the interpolation conditions
\begin{equation}
    \left\{
\begin{aligned}
&\Pi\left(\tau_k,\theta_\mathrm{Rx,k},\theta_\mathrm{Tx,k}\right)=\operatorname{sign}(l_k), \ \ \  \text{if }\left\{\tau_k,\theta_\mathrm{Rx,k},\theta_\mathrm{Tx,k}\right\} \in \bm{\eta};\\
&\left|\Pi(\tau,\theta_\mathrm{Rx},\theta_\mathrm{Tx})\right|<1, \quad\quad\quad\quad\quad  \text{otherwise},
\end{aligned}\label{sucondnoiseless}
\right.
\end{equation}
\end{proposition}
then the solution $\hat{{\boldsymbol{H}}}_v$ of Equation~(\ref{optprobnoiseless}) is unique and $\hat{{\boldsymbol{H}}}_v={\boldsymbol{H}}_v$.

\begin{proof}
See Appendix A.
\end{proof}
The function $\Pi(\tau,\theta_\mathrm{Rx},\theta_\mathrm{Tx})$ defines a trigonometric polynomial in three variables. Extending the analysis of the duality in~\cite{TangG}, Equation~(\ref{sucondnoiseless}) enforces that $\Pi(\tau,\theta_\mathrm{Rx},\theta_\mathrm{Tx})$ achieves modulus $1$ when $\left\{\tau,\theta_\mathrm{Rx},\theta_\mathrm{Tx}\right\}\in \bm{\eta}$, which will be further discussed in Section \ref{sec:validationProposition}. We recall from \cite{Tsai,Yang1} the conditions under which a dual certificate satisfying the hypothesis of Proposition~\ref{prop:uniqueandoptimal} exists, hence the tightness of ANM~\cite{da2018tight}.
\begin{proposition}[Exact recovery \cite{Tsai}]\label{prop:exactRecovery}
Let $\Delta_{\min}(\kappa)\triangleq\min_{i\neq j}\min(|\kappa_i-\kappa_j|,1-|\kappa_i-\kappa_j|)$. Given the conditions:\\
C1) $N_r, N_t\geq257$ and $N_s\geq513$;\\
C2) The separation conditions $\Delta_{\min}\left(\frac{d \sin \left(\theta_{\mathrm{Rx}}\right)}{\lambda_c}\right)\geq\frac{1}{\lfloor\frac{N_r-1}{4}\rfloor}$, $\Delta_{\min}\left(\frac{d \sin \left(\theta_{\mathrm{Tx}}\right)}{\lambda_c}\right)\geq\frac{1}{\lfloor\frac{N_t-1}{4}\rfloor}$, $\Delta_{\min}(\frac{\tau}{NT_s})\geq\frac{1}{\lfloor\frac{N_s-1}{8}\rfloor}$ hold;\\
the solution $\hat{{\boldsymbol{H}}}_v$ of Equation~(\ref{optprobnoiseless}) is unique and $\hat{{\boldsymbol{H}}}_v={\boldsymbol{H}}_v$.
\end{proposition}

{Since Equation~\eqref{sucondnoiseless} holds under conditions C1 and C2 in Proposition~\ref{prop:exactRecovery}, the optimality and the uniqueness of the estimate can be ensured via Proposition~\ref{prop:uniqueandoptimal}. Essentially, the condition C1 is a sufficient condition to guarantee the tightness of the atomic norm; however, this condition has been empirically shown to be unnecessary in practice \cite{Tsai,Yang1}. The condition C2 ensures that the bandwidth and the number of transmit and receive antennas is sufficiently large to ensure that the TOAs, AODs, and AOAs of the $K+1$ paths, respectively, are sufficiently separated.}

\subsection{Error Rate in Noisy Scenario}\label{sec:errorRate}
In practice, the observations $\bm{Y}$ are corrupted by noise. The atomic norm denoiser is a convex estimator that {balances} between the atomic norm of the solution and its deviation from the noiseless measurement model in Equation~\eqref{signalmodelstacked}. In our context, {the} atomic norm denoiser of the virtual matrix $\bm{H}_v$ {is given by} 
\begin{equation}
\boldsymbol{\hat{{H}}}_v=\mathop{\arg\min}_{\boldsymbol{H}_v}\quad {\epsilon}{\|\boldsymbol{H}_v\|}_\mathcal{A} +\frac{1}{2}{\|\boldsymbol{Y}-g(\boldsymbol{H}_v)\boldsymbol{X}\|}^2_F,\label{optprobnoisy}
\end{equation}
where $\epsilon$ is a regularization parameter whose value can be selected according to the next proposition to guarantee near-optimal denoising rates.

\begin{proposition}[Near-optimal denoising rate]\label{errorboundG}
Given the independent, zero mean, complex Gaussian distributed noise, {\it i.e.,} $\boldsymbol{W}[i,j] \sim \mathcal{CN}({0},\sigma^2)$, if we set $\epsilon$ to
\begingroup
\allowdisplaybreaks[0]
\begin{align}
\epsilon={}&\frac{2\sigma\sqrt{\sum_{n=1}^{N_s}\sum_{g=1}^{G}\left(\sum_{t=1}^{N_t}\boldsymbol{X}[(n-1)N_t+t,g]\right)^2}}{(N_s+1)\sqrt{N_t}} \nonumber \\ 
&\;\times\sqrt{\log\left(2\pi(N_s+N_r+N_t)\log(N_s+N_r+N_t)\right)+ 1} \nonumber\\
&\;\times\left(1+\frac{1}{\log(N_s+N_r+N_t)}\right),
\label{epsilon}
\end{align}
\endgroup
the following upper bound on the EER holds on the solution $\hat{\bm{H}}_v$ of~\eqref{optprobnoisy},
\begin{equation}
    \mathbb{E}\left\{\left\|\left(g(\hat{\boldsymbol{H}}_v)-g({\boldsymbol{H}}_v)\right)\boldsymbol{X}\right\|^{2}_{F}\right\}\leq 2\epsilon\left\|\boldsymbol{H}_v\right\|_\mathcal{A}.\label{expectederrorrate}
\end{equation}
\end{proposition}
\begin{proof}
See Appendix B.
\end{proof}

The bound in Equation~\eqref{expectederrorrate} is obtained by selecting the smallest $\epsilon$ that is greater than or equal to the expectation of the dual norm of $\sum_{i=1}^{N_s}\boldsymbol{Q}_{i,1}^\mathrm{H}\boldsymbol{W}\boldsymbol{X}^{\mathrm{H}}\boldsymbol{Q}_{i,2}^\mathrm{H}$ (see Equation~(\ref{defdualnorm}) for the definition of the dual norm).
Note that, as the noise $\boldsymbol{W}$ is added {to} $g({\boldsymbol{H}}_v)\boldsymbol{X}$ in our measurement model instead of $\boldsymbol{H}_v$, the theoretical results derived in~\cite{bhaskar2013atomic} cannot be directly applied to our {setting}.

In contrast to~\cite{Tsai}, which selects the regularization parameter $ \epsilon\varpropto\sigma\sqrt{\left(\frac{N_s+1}{2}\right)^2N_rN_t\log\left(\left(\frac{N_s+1}{2}\right)^2N_rN_t\right)}$, Proposition \ref{errorboundG} takes the structure of the virtual channel matrix and the energy of the pilot signal into consideration. Numerical evidence of the benefits of the choice of $\epsilon$ {provided} in Proposition \ref{errorboundG} will be given in Section \ref{sec:validationProposition}. 

\subsection{Semidefinite Approximation of the Atomic Norm}\label{sec:semidefiniteApproximation}

Given a $(2r_1 - 1) \times (2r_2-1)$ matrix $\bm{U}$, we define by $\mathcal{T}_2(\bm{U})$ the $r_1 r_2 \times r_1 r_2$ 2-level block Toeplitz matrix \cite{Yang2} as
\begin{align}
\mathcal{T}_2(\bm{U}) ={}& \begin{bmatrix} \mathcal{T}(\bm{u}_0) & \mathcal{T}(\bm{u}_{1}) & \cdots & \mathcal{T}(\bm{u}_{r_1 - 1}) \\
\mathcal{T}(\bm{u}_{-1}) & \mathcal{T}(\bm{u}_0) & \cdots & \mathcal{T}(\bm{u}_{r_1 - 2}) \\
\vdots & \vdots & \ddots & \vdots \\
\mathcal{T}(\bm{u}_{-r_1 + 1}) & \mathcal{T}(\bm{u}_{-r_1 + 2}) & \cdots & \mathcal{T}(\bm{u}_{0})
\end{bmatrix}
\end{align}
where  $\mathcal{T}(\bm{u})$ is defined for any a $2 r_2 - 1$ vector $\bm{u}$ as the $r_2 \times r_ 2$ Toeplitz matrix with first row $\left[\bm{u}[0], \bm{u}[1],\dots, \bm{u}[r_2 -1]\right]$ and first column $\left[\bm{u}[0], \bm{u}[1],\dots, \bm{u}[-r_2 + 1]\right]^{\mathrm{T}}$.

Although Equations~\eqref{optprobnoiseless} and~\eqref{optprobnoisy} are both convex programs involving multi-dimensional polynomial constraints, computing their optimal solution requires {resolving} a Lasserre hierarchy of semidefinite programs (SDP)~\cite{lasserre2001global}. In fact, it can be shown that the first Lasserre approximation is tight provided that the number of paths is small enough, which is the subject of Proposition~\ref{prop:equivalenceOfTheAtomicNorm}.
\begin{proposition}[Semidefinite representation]\label{prop:equivalenceOfTheAtomicNorm} 
The atomic norm  ${\|\boldsymbol{H}_v \|}_{\mathcal{A}}$ defined in Equation~(\ref{defatomicnorm}) is equivalent to the SDP
\begin{equation}
\begin{aligned}
\min_{\boldsymbol{V},\boldsymbol{U},\boldsymbol{H}_v}\quad{}& \frac{1}{2}\operatorname{Tr}\left(\boldsymbol{J}\right) \\
\text {s.t.}\quad{}&\boldsymbol{J}\triangleq\left[\begin{array}{cc}
\mathcal{T}_2({{\boldsymbol{U}}}) & {\boldsymbol{H}_v} \\
\boldsymbol{H}_v^{\mathrm{H}} & \mathcal{T}_2({{\boldsymbol{V}}})
\end{array}\right] \succeq \mathbf{0},\\
&\boldsymbol{H}_v^{(i,j)}=\boldsymbol{H}_v^{(k,z)},\textrm{ if } i+j=k+z, \forall i, j, k, z,
\end{aligned}\label{atomicnorm}
\end{equation}
provided that $\operatorname{rank}(\mathcal{T}_2(\hat{\boldsymbol{U}}))<\min\left(\frac{N_s+1}{2},N_r\right)$ and $\operatorname{rank}(\mathcal{T}_2(\hat{\boldsymbol{V}}))<\min\left(\frac{N_s+1}{2},N_t\right)$, 
where $\hat{\boldsymbol{U}}$ and $\hat{\boldsymbol{V}}$ are the estimates for ${\boldsymbol{U}}$ and ${\boldsymbol{V}}$, respectively.\label{equivalence}
\end{proposition}
\begin{proof}
Let $\|{\boldsymbol{H}_v}\|$ be the objective value in Equation~(\ref{atomicnorm}) and define $\operatorname{SDP}(\boldsymbol{H}_v)$ as {in}~\cite[Equation (35)]{Tsai} {that is,}
\begin{equation}
\begin{aligned}
\operatorname{SDP}(\boldsymbol{H}_v)\triangleq\min_{\boldsymbol{V},\boldsymbol{U},\boldsymbol{H}_v}\quad{}& \frac{1}{2}\operatorname{Tr}\left(\boldsymbol{J}\right) \\
\text {s.t.}\quad&\boldsymbol{J}\triangleq\left[\begin{array}{cc}
\mathcal{T}_2({{\boldsymbol{U}}}) & {\boldsymbol{H}_v} \\
\boldsymbol{H}_v^{\mathrm{H}} & \mathcal{T}_2({{\boldsymbol{V}}})
\end{array}\right] \succeq \mathbf{0}.
\end{aligned}\label{SDP}
\end{equation}
The inequality $\|\boldsymbol{H}_v\| \geq \operatorname{SDP}(\boldsymbol{H}_v)$ holds based on the definitions of the key quantities. It can be shown from~\cite[Lemma 1]{Tsai} that $\|\boldsymbol{H}_v\|\leq {\|\boldsymbol{H}_v\|}_{\mathcal{A}}$. Furthermore, given $\operatorname{rank}(\mathcal{T}_2(\hat{\boldsymbol{U}}))<\min\left(\frac{N_s+1}{2},N_r\right)$ and $\operatorname{rank}(\mathcal{T}_2(\hat{\boldsymbol{V}}))<\min\left(\frac{N_s+1}{2},N_t\right)$, $\mathcal{T}_2(\hat{\boldsymbol{U}})$ and $\mathcal{T}_2(\hat{\boldsymbol{
V}})$ admit the unique 2-level Vandermonde decomposition~\cite{Yang2}. Then,~\cite[Theorem 3]{Yang1} can be applied to the two-dimensional case, which, combined with~\cite[Lemma 2]{Tsai} yields the equality $\operatorname{SDP}(\boldsymbol{H}_v)={\|\boldsymbol{H}_v\|}_{\mathcal{A}}$, yielding on the desired result. 
\end{proof}

While the 2-level Vandermonde decompositions of the solutions $\mathcal{T}_2(\hat{\boldsymbol{U}})$ and  $\mathcal{T}_2(\hat{\boldsymbol{V}})$ of Equation~\eqref{SDP} need to exist for the SDP equivalence to hold, the rank conditions are simple sufficient conditions to {ensure} their existence. {In addition, we remark that the conditions C1 and C2 in Proposition \ref{prop:exactRecovery} are not necessary for the equivalence of $\operatorname{SDP}(\boldsymbol{H}_v)$ and ${\|\boldsymbol{H}_v\|}_{\mathcal{A}}$ to hold.}  

{Note that, a Toeplitz-Hankel formulation is proposed in \cite{Cho} for the recovery of one-dimensional signals, which is shown to be equivalent to the atomic norm when the Hankel matrix therein admits a Vandermonde decomposition. Though the formulation in \cite{Cho} might be extended for the multi-dimensional case, their extended formulation still does not fit our signal model since $\boldsymbol{H}_v$ is not a Hankel matrix. }

\subsection{Estimation of Individual Channel Parameters}\label{sec:channelParametersEstimation} Proposition \ref{prop:uniqueandoptimal} suggests that it is possible to estimate the TOAs, AOAs, and AODs ($\hat\tau_k$, $\hat\theta_{\mathrm{Tx},k}$, and $\hat\theta_{\mathrm{Rx},k}$) by identifying the values where the modulus of the dual trigonometric polynomial achieves {unity}. For computational efficiency, we estimate these channel parameters from the 2-level Vandermonde decomposition of the solution of Equation~(\ref{optprobnoiseless}) or~(\ref{optprobnoisy}) in the sequel. 
Specifically, the individual channel parameters are estimated via the matrix pencil and pairing (MaPP) algorithm~\cite{Yang2}{. The} estimates corresponding to the same path are paired, which presupposes the total number of the paths to be known by the receiver. If the total number of the paths $K$ is unknown by the receiver, it can be estimated by the rank of the solution $\widehat{\boldsymbol{H}}_v$ of the program in  Equation~(\ref{optprobnoiseless}) or~(\ref{optprobnoisy}) which can be computed regardless of the model order. In practice, it is {sufficient} to count the number of eigenvalues of $\boldsymbol{\hat{{H}}}_v$ that are greater than a threshold $\varrho$ to get a reliable estimate of $K$ before applying the MaPP algorithm.
We assume that the channel coefficient $\gamma_0$ associated to the LOS path is the one with largest modulus when $K$ is to be estimated. The choice of $\varrho$ is empirically discussed in Section \ref{subsec:modelOrder}.

 
\subsection{Localization and Orientation Estimation}\label{sec:EXIP}

Although the estimated location and orientation can be directly estimated from the geometry of the LOS path, more accurate estimates can be achieved by leveraging the geometry of the NLOS paths~\cite{Shahmansoori}. Once the parameter $\boldsymbol{\eta}$, which parametrizes Equation~(\ref{optprobnoisy}) given the channel coefficients $\{\gamma_{0},\gamma_{1},\cdots\,\gamma_{K}\}$, is estimated through the procedure presented in Section \ref{sec:channelParametersEstimation}, the final step consists of recovering the location and orientation parameters from the geometric mapping in Equation~\eqref{eq:geometricMapping}.

Since we make no assumptions on the path loss model in the signal model, knowledge of the channel coefficients do not improve the accuracy of the localization and orientation estimation. In addition, $\{\boldsymbol p,\theta_o,\boldsymbol{s}_1, \boldsymbol{s}_2,\cdots,\boldsymbol{s}_K, \gamma_0,\gamma_1,\cdots,\gamma_K\}$ can be used to re-parametrize the optimization problem in Equation~(\ref{optprobnoisy}). Therefore, we fix the estimated channel coefficients\footnote{We can substitute the estimated $\boldsymbol{\eta}$ into Equation~(\ref{signalmodelstacked}) to achieve a system of linear equations to compute the estimates of channel coefficients \cite{Li}.} and propose a weighted least squares problem to achieve an accurate localization and orientation estimation, with the estimates of all the paths, {\it i.e.,} $\hat{\boldsymbol\eta}$, exploited,
\begin{equation}
{\hat{\tilde{\boldsymbol\eta}}} = \arg\min_{\tilde{\boldsymbol\eta}} \left(\hat{\boldsymbol\eta}-f({\tilde{\boldsymbol\eta}})\right)^{\mathrm{T}}\boldsymbol{\mathcal{D}}\left(\hat{\boldsymbol\eta}-f({\tilde{\boldsymbol\eta}})\right),
\label{exip_etatilde}
\end{equation}
where the mapping $f({\tilde{\boldsymbol\eta}})={\boldsymbol\eta}$ integrates the geometric mapping in Equation~\eqref{eq:geometricMapping}. Inspired by the EXIP~\cite{Stoica,Shahmansoori}, the weight matrix $\boldsymbol{\mathcal{D}}$ in Equation~\eqref{exip_etatilde} is set to the Hessian of the objective function $L(\boldsymbol{\eta})$ of the program in Equation~(\ref{optprobnoisy}) at the estimated channel parameters $\hat{\bm{\eta}}$, {\it i.e.}, 
\begin{equation}
    \begin{aligned}
        \boldsymbol{\mathcal{D}} \triangleq
        &\left[\begin{array}{cccc}
        \frac{\partial^2{L}(\hat{\boldsymbol\eta})}{\partial{\tau_0}\partial{\tau_0}}
        &\frac{\partial^2{L}(\hat{\boldsymbol\eta})}{\partial{\tau_0}\partial{\theta_{\mathrm{Tx},0}}} & \cdots&\frac{\partial^2{L}(\hat{\boldsymbol\eta})}{\partial{\tau_0}\partial{\theta_{\mathrm{Rx},L}}}\\
        \frac{\partial^2{L}(\hat{\boldsymbol\eta})}{\partial{\theta_{\mathrm{Tx},0}}\partial{\tau_0}}& \frac{\partial^2{L}(\hat{\boldsymbol\eta})}{\partial{\theta_{\mathrm{Tx},0}}\partial{\theta_{\mathrm{Tx},0}}}& \cdots & \frac{\partial^2{L}(\hat{\boldsymbol\eta})}{\partial{\theta_{\mathrm{Tx},0}}\partial{\theta_{\mathrm{Rx},L}}}\\
        \vdots & \vdots & \ddots & \vdots \\
        \frac{\partial^2{L}(\hat{\boldsymbol\eta})}{\partial{\theta_{\mathrm{Rx},L}}\partial{\tau_0}}& \frac{\partial^2{L}(\hat{\boldsymbol\eta})}{\partial{\theta_{\mathrm{Rx},L}}\partial{\theta_{\mathrm{Tx},0}}}& \cdots & \frac{\partial^2{L}(\hat{\boldsymbol\eta})}{\partial{\theta_{\mathrm{Rx},L}}\partial{\theta_{\mathrm{Rx},L}}}\\
        \end{array}\right],
    \end{aligned} \label{eq:Dmatrix}
\end{equation}
which depends on the channel parameters estimated via the proposed ANM based method of Section \ref{sec:errorRate}\footnote{We note that the derivative of the term $\epsilon\left\Vert \bm{H}_v \right\Vert_{\mathcal{A}}$ in Equation~\eqref{optprobnoisy} with respect to the parameters in $\bm{\eta}$ is $0$, so Equation~\eqref{eq:Dmatrix} reduces to the Hessian of the quadratic loss function $\frac{1}{2}{\|\boldsymbol{Y}-g(\boldsymbol{H}_v)\boldsymbol{X}\|}^2_F$ at $\hat{\bm{\eta}}$.}.

The non-linear least squares problem in Equation~(\ref{exip_etatilde}) can be solved via the Levenberg-Marquard-Fletcher algorithm~\cite{Fletcher}. The parameters in $\tilde{\boldsymbol{\eta}}$ are initialized with the values $\boldsymbol{\hat p}_{\text{LOS}}$, $\hat{\theta}_{o,\text{LOS}}$, $\{\hat s_{1,y,\text{LOS}},\hat s_{2,y,\text{LOS}},\cdots,\hat s_{K,y,\text{LOS}}\}$, and $\{\hat s_{1,x,\text{LOS}},\hat s_{2,x,\text{LOS}},\cdots,\hat s_{K,x,\text{LOS}}\}$, which are derived in the following set of equations,
\begingroup
\allowdisplaybreaks
\begin{subequations}\label{eq:locEstimator}
\begin{align}
\boldsymbol{\hat p}_{\text{LOS}} &= \boldsymbol{q} + c\hat\tau_0[\cos({\hat\theta}_{\mathrm{Tx},0}),\sin(\hat\theta_{\mathrm{Tx},0})]^{\mathrm{T}},\label{hatp}\\
\hat{\theta}_{_o,\text{LOS}} &= \pi + \hat\theta_{\mathrm{Tx},0} - \hat\theta_{\mathrm{Rx},0},\\
\hat s_{k,y,\text{LOS}} &= \tan(\hat\theta_{\mathrm{Tx},k})(\hat s_{k,x}-q_{x})+q_{y},\\
\hat{s}_{k,x,\text{LOS}} &={} \nonumber\\
& \hspace{-.3in} \frac{\tan(\hat\theta_{\mathrm{Tx},k})q_{x}-\tan(\hat\theta_{\mathrm{Rx},k}+\hat{\theta}_{_o,\text{LOS}})\hat p_{\text{LOS},x}+\hat p_{\text{LOS},y}-q_{y}}{\tan(\hat\theta_{\mathrm{Tx},k})-\tan(\hat\theta_{\mathrm{Rx},k}+\hat{\theta}_{_o,\text{LOS}})}.
\label{hatsky}
\end{align}
\end{subequations}
\endgroup
Note that, Equation~(\ref{exip_etatilde}) implicitly depends on the received signal via $\hat{\boldsymbol\eta}$ and $\boldsymbol{\mathcal{D}}$. Essentially, the non-linear weighted least squares problem in Equation~(\ref{exip_etatilde}) is a second-order Taylor expansion of the cost function in Equation~(\ref{optprobnoisy}) re-parametrized by $\{\boldsymbol p,\theta_o,\boldsymbol{s}_1, \boldsymbol{s}_2,\cdots,\boldsymbol{s}_K, \gamma_0,\gamma_1,\cdots,\gamma_K\}$ at the global optimum~\cite{Stoica}, which is different from more typical {\em indirect localization} methods with multilateration \cite{Zekavat}.

\begin{figure*}[t]
\centering
\includegraphics[scale=0.54]{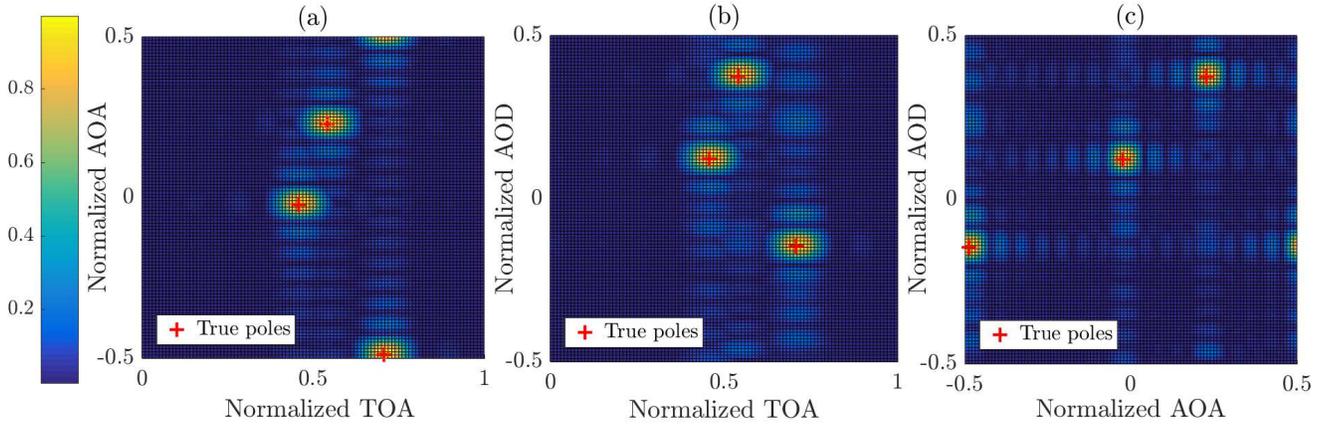}
\caption{{The modulus of the dual polynomial in the absence of noise, defined as the solution of the dual problem in Equation~\eqref{dualoptprob}.}}
\label{fig:dualpoly}
\end{figure*}
\section{Simulation Results}\label{sec:sim}

In this section, we evaluate the performance of our proposed scheme {, {\it i.e.,} LOCMAN.} First, the theoretical analyses presented in Propositions~\ref{prop:uniqueandoptimal}-\ref{errorboundG} are numerically validated in Section \ref{sec:validationProposition}. Then, in Section~\ref{sec:accuracyComparison}, our LOCMAN scheme is compared with DCS-SOMP~\cite{Duarte,Shahmansoori} to the show the performance improvements over the estimation accuracy of individual channel parameters, orientation and location. Finally, the effects of the incorrect model order selection as well as synchronization errors are studied via simulation in Section~\ref{subsec:modelOrder}.

\subsection{Signal Parameters}
In all of the numerical results, we set $f_c = \SI{60}{\giga\Hz}$, $B=\SI{100}{\mega\Hz}$, $N_s=15$, $N_r=16$, $N_t=16$, $G=16$, $K=2$, and $d=\frac{\lambda_c}{2}$. The pilot signals are random complex values uniformly distributed on the unit circle. Unless otherwise stated, the total number of paths is known at the receiver and the channel coefficients are generated based on the free-space path loss model~\cite{Goldsmith} in the simulation. The BS is placed at $[\SI{0}{\meter},\SI{0}{\meter}]^{\mathrm{T}}$ while the target is at $[\SI{20}{\meter},\SI{5}{\meter}]^{\mathrm{T}}$ with an orientation $\theta_o=\SI{0.2}{\radian}$. The scatterers corresponding to two NLOS paths are placed at $[\SI{7.45}{\meter},\SI{8.54}{\meter} ]^{\mathrm{T}}$ and $[\SI{19.89}{\meter}, \SI{-6.05}{\meter}]^{\mathrm{T}}$.  We solve the optimization problem in Equations~(\ref{optprobnoiseless}) and~(\ref{optprobnoisy}) for the channel estimation in the noiseless and noisy cases, {where the semidefinite representation of the atomic norm in Proposition \ref{prop:equivalenceOfTheAtomicNorm} is adopted}. For the atomic norm denoiser in Equation~\eqref{optprobnoisy}, the regularization parameters $\epsilon$ is set according to Proposition  \ref{errorboundG}. In the presence of  independent, zero mean, complex Gaussian distributed noise, the SNR is defined as $\text{SNR} = \frac{\left\Vert \bm{H} \bm{X} \right\Vert_{F}^2}{  \left\Vert \bm{W} \right\Vert_{F}^2}$.

\subsection{Validation of Propositions \ref{prop:uniqueandoptimal}-\ref{errorboundG}}\label{sec:validationProposition}

\begin{figure}[t]
\centering
\includegraphics[scale=0.52]{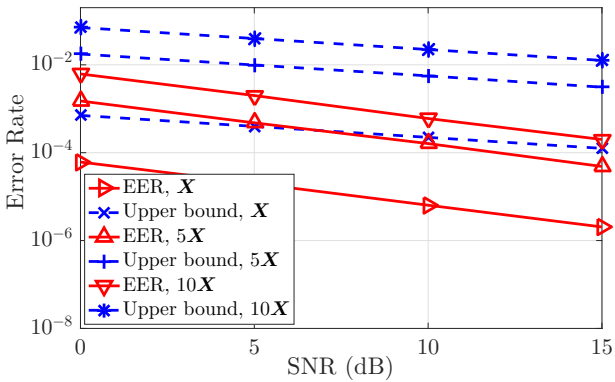}\vspace{-1pt}
\caption{{The EER of the estimate and its upper bound.}}
\label{fig:expectedmse}
\end{figure}

Fig. \ref{fig:dualpoly} shows the modulus of the dual polynomial for the noiseless case. Recall that, the dual polynomial is a  trigonometric  function with respect to TOAs, AOAs, and AODs. For better visualization, the modulus of the dual polynomial is projected on the TOA-AOA, TOA-AOD, and AOA-AOD axes in Fig. \ref{fig:dualpoly} (a), (b), and (c), respectively. Consistent with our analysis for Proposition \ref{prop:uniqueandoptimal}, from Fig. \ref{fig:dualpoly}, we can observe that the modulus of dual polynomial nearly equals $1$ when $\{\tau,\theta_\mathrm{Rx},\theta_\mathrm{Tx}\}\in \bm{\eta}$ {even} though condition C1 of Proposition \ref{prop:exactRecovery} is not satisfied given the values of $N$, $N_r$, and $N_t$  selected in our experiments.

For the noisy scenario, the EER of the estimate and its theoretical bound are shown in Fig. \ref{fig:expectedmse}. We can verify that that $2\epsilon\left\|\boldsymbol{H}_v\right\|_\mathcal{A}$ is an upper bound of $\mathbb{E}\left\{\left\|\left(g(\hat{\boldsymbol{H}}_v)-g({\boldsymbol{H}}_v)\right)\boldsymbol{X}\right\|^{2}_{F}\right\}$, as stated in Proposition \ref{errorboundG}. Though the bound is not sharp, Proposition \ref{errorboundG} suggests a good choice of the regularization parameter $\epsilon$. To show the benefits of setting $\epsilon$ according to Equation~(\ref{epsilon}), we present in Fig.~\ref{fig:choiceofepsilon} the RMSEs of localization, with different choices of pilot signals, {\it i.e.,} $\boldsymbol{X}$, $5\boldsymbol{X}$, $10\boldsymbol{X}$, when $\epsilon$ is selected according to Proposition~\ref{errorboundG} and according to~\cite{Tsai} for comparison. We observe that the localization accuracy using our design of $\epsilon$ yield near constant RMSEs as the energy of pilot signal increases while the performance degrades with the choice suggested in~\cite{Tsai}. This strong improvement can be explained by the fact that  our proposed choice of the regularization parameter $\epsilon$ incorporates the energy of the pilot signals and the structure of the virtual channel matrix, leading to a dynamic penalization of the atomic norm in Equation~\eqref{optprobnoisy}.

\begin{figure}[t]
\centering
\includegraphics[scale=0.52]{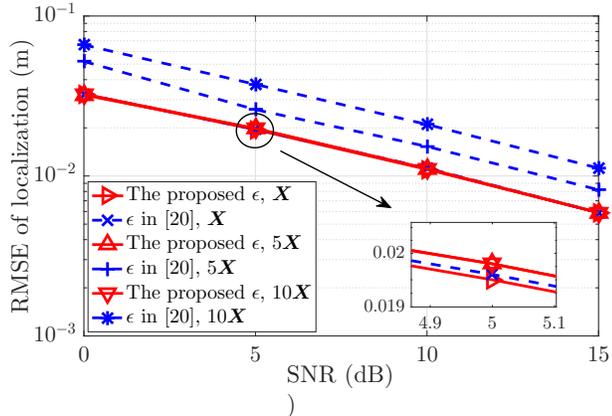}
\caption{{The influence of the choice of the pilot signal on the localization accuracy.}}
\label{fig:choiceofepsilon}
\end{figure}


\subsection{Estimation Accuracy Comparison}\label{sec:accuracyComparison}

The RMSEs of TOA, AOA, and AOD estimation using our LOCMAN scheme are shown in Figs. \ref{fig:channelestimation} (a)-(c), respectively, with comparisons to the DCS-SOMP based method~\cite{Duarte,Shahmansoori} and to the Cram{\'e}r-Rao lower bounds (CRLBs)~\cite{Shahmansoori}. As observed in Fig.~\ref{fig:channelestimation}, our LOCMAN scheme outperforms the DCS-SOMP based method, as it relies on a discretization of the parameter space, and suffers from basis mismatch.
In contrast, the estimation accuracy of LOCMAN does not rely on a grid resolution. Furthermore, Figs.~\ref{fig:channelestimation} (a)-(c) shows that the RMSEs of TOA, AOA, and AOD estimation nearly coincide with the CRLB.

\begin{figure}[t]
\centering
\includegraphics[scale=0.52]{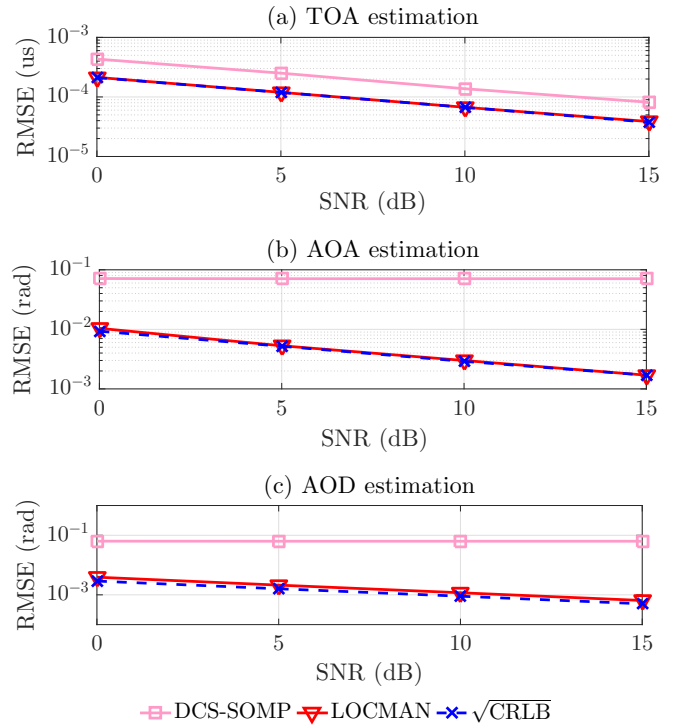}
\caption{{(a) The RMSE of TOA estimation; (b) The RMSE of AOA estimation; (c) The RMSE of AOD estimation.}}
\label{fig:channelestimation}
\end{figure}
\begin{figure}[t]
\centering
\includegraphics[scale=0.52]{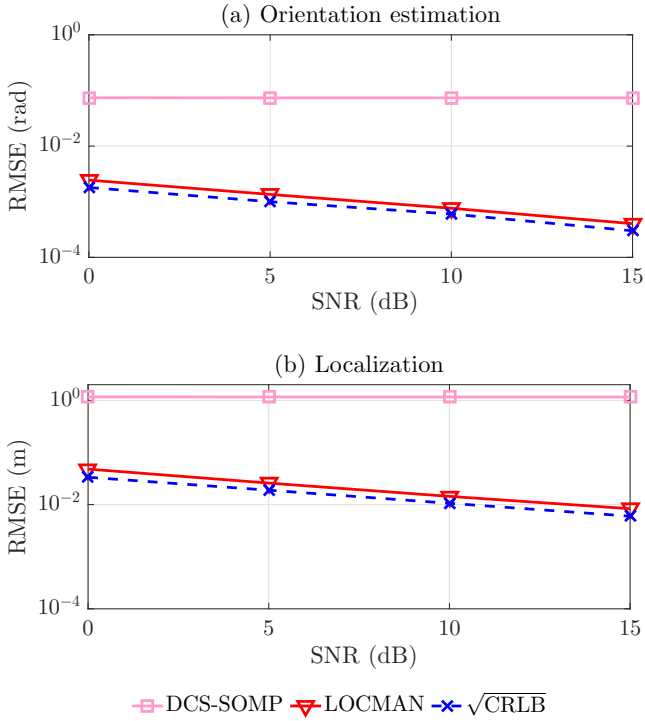}\vspace{-2.5pt}
\caption{{(a) The RMSE of orientation estimation; (b) The RMSE of localization.}}
\label{fig:localization}
\end{figure}

\begin{figure}[t]
\centering
\includegraphics[scale=0.52]{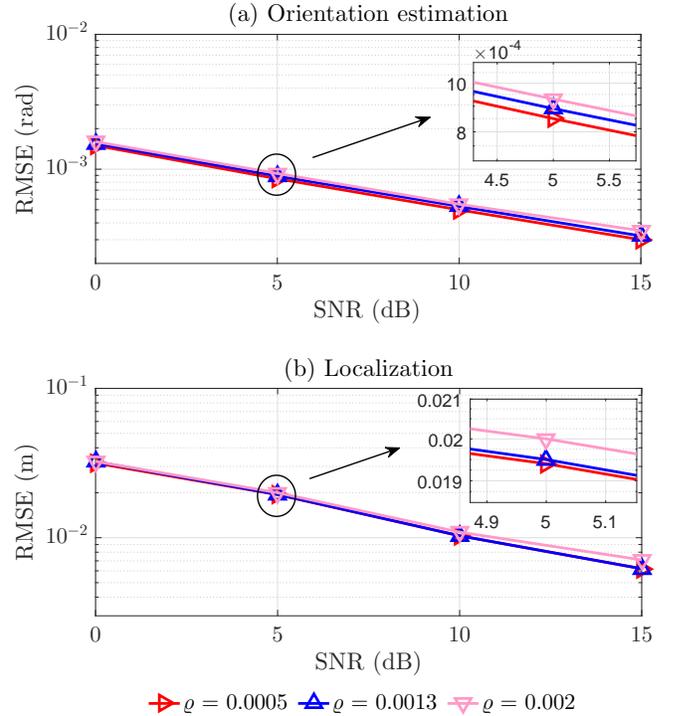}\vspace{-3pt}
\caption{{The influence of incorrect model order selection on localization and orientation estimation.}}
\label{fig:modelorder}
\end{figure}

Due to the quality of our super-resolution channel estimation, lower RMSEs for localization and orientation estimation are achieved as seen in Figs. \ref{fig:localization} (a) and (b) versus the DCS-SOMP based method~\cite{Shahmansoori}\footnote{For comparison fairness, the refinement of estimates of channel parameters in~\cite{Shahmansoori} is not implemented for both schemes. Note that, as compared 
with the DCS-SOMP based method, our scheme could provide more accurate estimates for the initialization of the refinement stage to avoid local optima.}. In the SNR sense, more than $7$ dB gain are obtained using our scheme. In addition, there is only around 2dB gap between the RMSE of localization or orientation estimation using our LOCMAN scheme and the CRLB curves, verifying the efficacy of our design. 

\subsection{Model Order Selection and Synchronization Error}\label{subsec:modelOrder}

We next study the scenario where the total number of paths $K$ is unknown at the receiver. The efficiency of the thresholding-based model order selection proposed in Section \ref{sec:method}-D is investigated. To show the effect of the choices of $\varrho$ on the performance degradation, we fix $l_0$, $l_1$, and $l_2$ in Equation~(\ref{virtualchannel}) to $0.0025$ $0.0011$, and $0.0014$, while we vary the value of $\varrho$. Figure \ref{fig:modelorder} presents the RMSE of orientation estimation and localization for $\varrho=0.0005$, $\varrho=0.0013$ and $\varrho=0.002$, respectively, from which a similar estimation accuracy for different choices of $\varrho$ can be observed. Since we assume that the LOS path is that with the largest modulus of channel coefficient in Section \ref{sec:method}-A when $K$ is unknown at the receiver, the channel parameters associated with the LOS path are always available for localization\footnote{This is achievable unless the value of $\varrho$ is so large that the estimated total number of paths is 0.} and the loss caused by the estimation error of the model order is less than 1 dB with respect to the RMSE of localization or orientation estimation.

Finally, the performance degradation caused by synchronization errors is investigated though synchronization errors are not considered in our signal model. Specifically, denoting by $\tilde{\tau}$ the unknown synchronization error, we define AOAs and AODs according to Equations~(\ref{thetatx0})-(\ref{thetarxk}) while re-express the TOA of each path as 
\begin{subequations}
\begin{align}
\tau_{0} &=\frac{\left\|\boldsymbol{p}-\boldsymbol{q}\right\|_{2}} { c}+\tilde{\tau},\label{tau0se}\\
\tau_{k} &=\frac{\left\|\boldsymbol{q}-\boldsymbol{s}_{k}\right\|_{2} +\left\|\boldsymbol{p}-\boldsymbol{s}_{k}\right\|_{2}} {c}+\tilde{\tau}, \quad k>0.\label{taukse} 
\end{align}
\end{subequations} 
We still exploit the LOCMAN proposed in Section \ref{sec:method} for channel estimation as well as localization. The RMSEs of orientation estimation and localization are presented in Figs. \ref{fig:SE} (a) and (b), respectively. Since the definitions of AOAs and AODs in Equations~(\ref{thetatx0})-(\ref{thetarxk}) are maintained, the synchronization error has negligible influence on the AOA or AOD estimation and correspondingly, the RMSE of orientation estimation is nearly constant with the increase of the value of $\tilde{\tau}$, as observed in Fig. \ref{fig:SE} (a). In addition, the synchronization error has a minor effect on the TOA estimation as the value of $\tilde{\tau}$ is quite small as compared with the TOA of each path in the simulation\footnote{Note that, for a localization problem, the synchronization error is generally at ns-level\cite{Zhao}.}. In contrast, the synchronization error degrades the localization accuracy, according to Fig. \ref{fig:SE} (b). However, as compared with the DCS-SOMP based method, the centimeter-level localization accuracy is still attained using our scheme when $\tilde{\tau}\leq \SI{1E-3}{\micro\second}$. The estimation design of the synchronization error is left as the future work.

\begin{figure}[t]
\centering
\includegraphics[scale=0.52]{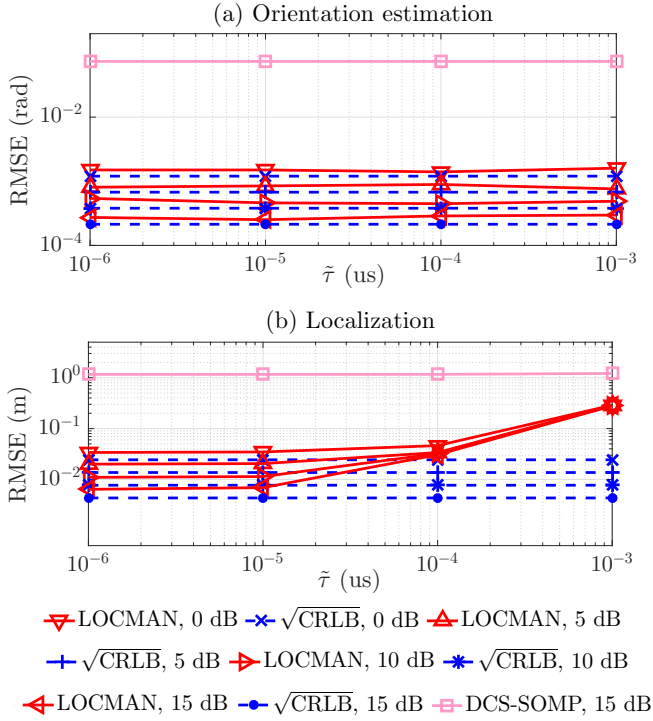}
\caption{{The influence of synchronization error on localization and orientation estimation.}}
\label{fig:SE}
\end{figure}

\section{Conclusions}\label{sec:con}
In this paper, a multi-dimensional atomic norm based method is proposed for  high-accuracy localization and orientation estimation in mmWave MIMO OFDM systems.  To effectively estimate all of the location-relevant channel parameters with super-resolution, a novel virtual channel matrix is introduced and an atomic norm based channel estimator is proposed. Theoretical performance guarantees are derived for both the noiseless and noisy cases. Using the estimates of all the paths, a weighted least squares problem is proposed based on the extended invariance principle to accurately recover the location and orientation. The new method offers strong improvements with respect to the RMSE of estimation over prior work~\cite{Shahmansoori} (more than $7$ dB gain) and exhibits high robustness in terms of the incorrect model order selection or synchronization errors. Furthermore, with the proposed method, the RMSEs of channel estimation, localization and orientation estimation nearly achieve  the CRLBs.

\appendices
\begingroup
\allowdisplaybreaks

\section{Proof of Proposition \ref{prop:uniqueandoptimal}}
To derive the optimality and uniqueness conditions, we follow the proof structure of \cite{candes2014towards,TangG}, and define the dual norm as
\begin{equation}
\begin{aligned}
\| \boldsymbol{Z}\|_{\mathcal{A}}^{*}&=\sup _{\|\boldsymbol{\tilde{{H}}}_v\|_{\mathcal{A}}\leq 1} \langle \boldsymbol{Z}  , \boldsymbol{\tilde{{H}}}_v\rangle_R\\
&= \sup_{ \tau,\theta_\mathrm{Rx},\theta_\mathrm{Tx},\phi} \langle \boldsymbol{Z}  , \boldsymbol{A}(\tau,\theta_\mathrm{Rx},\theta_\mathrm{Tx},\phi)\rangle_R,\\
&= \sup_{ \tau,\theta_\mathrm{Rx},\theta_\mathrm{Tx}} \left|\langle \boldsymbol{Z}  , \boldsymbol{A}(\tau,\theta_\mathrm{Rx},\theta_\mathrm{Tx},0)\rangle\right|.
\end{aligned}\label{defdualnorm}
\end{equation}
The Lagrange dual problem of Equation~(\ref{optprobnoiseless})  is given by
\begin{equation}
\begin{aligned}
\mathop{\arg\min}_{{\boldsymbol{\Lambda}}}\quad &\langle {\boldsymbol{\Lambda}}, \boldsymbol{Y}\rangle_{R} \\
\text { s.t. }\quad&\left\|\sum_{i=1}^{{N_s}{}}\boldsymbol{Q}_{i,1}^\mathrm{H} {\boldsymbol{\Lambda}}{\boldsymbol{X}}^\mathrm{H} \boldsymbol{Q}_{i,2}^\mathrm{H}\right\|_{\mathcal{A}}^{*}\leq 1.
\end{aligned}
\label{dualoptprob}
\end{equation}
It can be verified that any $\boldsymbol{\Lambda}$ satisfying Equation~(\ref{sucondnoiseless}) is dual feasible. Next, we have that
\begingroup
\allowdisplaybreaks
\begin{align}
&\|\boldsymbol{H}_v\|_\mathcal{A}\overset{\mathrm{(a)}}{\geq}\| \boldsymbol{H}_{v}\|_\mathcal{A}\left\| \sum_{i=1}^{N_s}\boldsymbol{Q}_{i,1}^{\mathrm{H}} {\boldsymbol{\Lambda}}{\boldsymbol{X}}^\mathrm{H}\boldsymbol{Q}_{i,2}^{\mathrm{H}} \right\|_\mathcal{A}^{*} \nonumber\\
&\overset{\mathrm{(b)}}{\geq} \left\langle\sum_{i=1}^{N_s}\boldsymbol{Q}_{i,1}^{\mathrm{H}} {\boldsymbol{\Lambda}}{\boldsymbol{X}}^\mathrm{H}\boldsymbol{Q}_{i,2}^{\mathrm{H}}, \boldsymbol{H}_{v}\right\rangle_{R} \nonumber\\
&=\left\langle\sum_{i=1}^{N_s}\boldsymbol{Q}_{i,1}^{\mathrm{H}} {\boldsymbol{\Lambda}}{\boldsymbol{X}}^\mathrm{H}\boldsymbol{Q}_{i,2}^{\mathrm{H}}, \sum_{k=0}^{K} l_{k} \boldsymbol{A}(\tau_k,\theta_\mathrm{Rx,k},\theta_\mathrm{Tx,k},0)\right\rangle_{R} \nonumber\\
&=\sum_{k=0}^{K} \operatorname{Re}\left(l_{k}^{*} \left\langle\sum_{i=1}^{N_s} \boldsymbol{Q}_{i,1}^{\mathrm{H}} {\boldsymbol{\Lambda}}{\boldsymbol{X}}^\mathrm{H}\boldsymbol{Q}_{i,2}^{\mathrm{H}}, \boldsymbol{A}(\tau_k,\theta_\mathrm{Rx,k},\theta_\mathrm{Tx,k},0)\right\rangle\right) \nonumber\\
&=\sum_{k=0}^{K}\operatorname{Re}\left( l_{k}^{*}\Pi(\tau_k,\theta_\mathrm{Rx,k},\theta_\mathrm{Tx,k},0)\right)\overset{\mathrm{(c)}}{=}\sum_{k=0}^{K}\left|l_{k}\right|\geq\left\|\boldsymbol{H}_{v}\right\|_\mathcal{{A}},
\end{align}
\endgroup
where $\mathrm{(a)}$ and $\mathrm{(c)}$ hold because of Equation~(\ref{sucondnoiseless}) while $\mathrm{(b)}$ results from  H{\"o}lder inequality. Therefore, $ \langle {\boldsymbol{\Lambda}}, \boldsymbol{Y}\rangle_{R}=\left\|\boldsymbol{H}_{v}\right\|_\mathcal{{A}}$ and strong duality holds. Then, $\boldsymbol{H}_{v}$ is a primal optimal solution and ${\boldsymbol{\Lambda}}$ is a dual optimal solution.

We show the uniqueness by contradiction by assuming the existence of a distinct optimal solution $\tilde{\boldsymbol{H}}_{v}$ such that 
\begin{equation}
\tilde{\boldsymbol{H}}_{v}=\sum_{\tilde{k}=0}^{\tilde{K}} \tilde{l}_{k} \boldsymbol{A}\left(\tilde{\tau}_k,\tilde{\theta}_\mathrm{Rx,k},\tilde{\theta}_\mathrm{Tx,k},0\right)
\end{equation}
with
$
{\|\tilde{\boldsymbol{H}}_v\|}_\mathcal{A}=\sum_{\tilde{k}=0}^{\tilde{K}}|\tilde{l}_{k}|
$.
From Equation~(\ref{sucondnoiseless}), we have Equation~(\ref{uniqueness}), which contradicts the strong duality. Hence, the optimal solution is unique, concluding the proof.
\begin{figure*}[!b]
\hrule
\begin{align}
\left\langle\sum_{i=1}^{N_s}\boldsymbol{Q}_{i,1}^{\mathrm{H}} {\boldsymbol{\Lambda}}{\boldsymbol{X}}^\mathrm{H}\boldsymbol{Q}_{i,2}^{\mathrm{H}}, \tilde{\boldsymbol{H}}_{v}\right\rangle_{R}={}&\left\langle\sum_{i=1}^{N_s}\boldsymbol{Q}_{i,1}^{\mathrm{H}} {\boldsymbol{\Lambda}}{\boldsymbol{X}}^\mathrm{H}\boldsymbol{Q}_{i,2}^{\mathrm{H}}, \sum_{\tilde{k}=0}^{\tilde{K}} \tilde{l}_{\tilde{k}} \boldsymbol{A}\left(\tilde{\tau}_{\tilde{k}},\tilde{\theta}_\mathrm{Rx,\tilde{k}},\tilde{\theta}_\mathrm{Tx,\tilde{k}},0\right)\right\rangle_{R} \nonumber \\
={}&\sum_{(\tau_{\tilde{k}},\theta_\mathrm{Rx,\tilde{k}},\theta_\mathrm{Tx,\tilde{k}}) \in \bm{\eta}} \operatorname{Re}\left( \tilde{l}_{\tilde{k}}^{*} \left\langle\sum_{i=1}^{N_s} \boldsymbol{Q}_{i,1}^{\mathrm{H}} {\boldsymbol{\Lambda}}{\boldsymbol{X}}^\mathrm{H}\boldsymbol{Q}_{i,2}^{\mathrm{H}}, \boldsymbol{A}(\tau_{\tilde{k}},\theta_\mathrm{Rx,\tilde{k}},\theta_\mathrm{Tx,\tilde{k}},0)\right\rangle\right) \nonumber\\
{}&\quad + \sum_{(\tau_{\tilde{k}},\theta_\mathrm{Rx,\tilde{k}},\theta_\mathrm{Tx,\tilde{k}}) \notin \bm{\eta}}\operatorname{Re}\left( \tilde{l}_{\tilde{k}}^{*} \left\langle\sum_{i=1}^{N_s} \boldsymbol{Q}_{i,1}^{\mathrm{H}} {\boldsymbol{\Lambda}}{\boldsymbol{X}}^\mathrm{H}\boldsymbol{Q}_{i,2}^{\mathrm{H}}, \boldsymbol{A}(\tau_{\tilde{k}},\theta_\mathrm{Rx,\tilde{k}},\theta_\mathrm{Tx,\tilde{k}},0)\right\rangle\right) \nonumber \\
<{}& \sum_{(\tau_{\tilde{k}},\theta_\mathrm{Rx,\tilde{k}},\theta_\mathrm{Tx,\tilde{k}}) \in\bm{\eta}} |\tilde{l}_{\tilde{k}}^{*}| +\sum_{(\tau_{\tilde{k}},\theta_\mathrm{Rx,\tilde{k}},\theta_\mathrm{Tx,\tilde{k}}) \notin\bm{\eta}} |\tilde{l}_{\tilde{k}}^{*}|=\|\tilde{\boldsymbol{H}}_v\|_\mathcal{A}.
\label{uniqueness}
\end{align}
\end{figure*}

\section{Proof of Proposition  \ref{errorboundG}}
For the proof of Proposition \ref{errorboundG}, we first extend~\cite[Lemma1 and Theorem 1]{bhaskar2013atomic} to our atomic norm denoising formulation.

\begin{lemma}\label{optimalitynoisy}
If $\hat{\boldsymbol{H}}_v$ is the solution of Equation~(\ref{optprobnoisy}), it holds that 
\begin{subequations}
\begin{align}
\left\|\sum_{i=1}^{N_s}\boldsymbol{Q}_{i,1}^\mathrm{H}\left(\boldsymbol{Y}-g(\hat{\boldsymbol{H}}_v)\boldsymbol{X}\right)\boldsymbol{X}^{\mathrm{H}}\boldsymbol{Q}_{i,2}^\mathrm{H}\right\|_{\mathcal{A}}^{*}\leq{}& \epsilon,\\
\left\langle\boldsymbol{Y}-g(\hat{\boldsymbol{H}}_v)\boldsymbol{X},g(\hat{\boldsymbol{H}}_v)\boldsymbol{X}\right\rangle_R={}&\epsilon\left\|\hat{{\boldsymbol{H}}}_v\right\|_{\mathcal{A}}.
\end{align}
\label{lemma1}
\end{subequations}
\end{lemma}

\begin{proof}
Define 
\(
a(\hat{\boldsymbol{H}}_v)\triangleq\left\|\boldsymbol{Y}-g(\hat{\boldsymbol{H}}_v)\boldsymbol{X}\right\|+\epsilon\left\|\hat{\boldsymbol{H}}_v\right\|_{\mathcal{A}}.
\)
The subdifferential $\partial a(\hat{\boldsymbol{H}}_v)$ of the previous writes
\begin{equation}
\partial a(\hat{\boldsymbol{H}}_v)=\sum_{i=1}^{N_s}\boldsymbol{Q}_{i,1}^\mathrm{H}\left(g(\hat{\boldsymbol{H}}_v)\boldsymbol{X}-\boldsymbol{Y}\right)\boldsymbol{X}^{\mathrm{H}}\boldsymbol{Q}_{i,2}^\mathrm{H}+\epsilon\boldsymbol{Z},
\end{equation}
where
$
\boldsymbol{Z}\in \partial\left\|\hat{\boldsymbol{H}}_v\right\|_\mathcal{A}.
$
Since $\hat{\boldsymbol{H}}_v$ is the solution of Equation~(\ref{optprobnoisy}), 
$
\boldsymbol{0}\in \partial a(\hat{\boldsymbol{H}}_v)
$ holds which leads to the equality 
\begin{equation}
\sum_{i=1}^{N_s}\boldsymbol{Q}_{i,1}^\mathrm{H}\left(\boldsymbol{Y}-g(\hat{\boldsymbol{H}}_v)\boldsymbol{X}\right)\boldsymbol{X}^{\mathrm{H}}\boldsymbol{Q}_{i,2}^\mathrm{H}=\epsilon \boldsymbol{Z}.\label{subdifferentialequal}
\end{equation}

Furthermore, the properties of the subdifferential of norms imply that
$\|\boldsymbol{Z}\|_{\mathcal{A}}^{*} \leq 1$, and 
$\left\langle\hat{\boldsymbol{H}}_v, \boldsymbol{Z}\right\rangle_{R}=\left\|\hat{\boldsymbol{H}}_v\right\|_\mathcal{A}$. One concludes on Lemma \ref{optimalitynoisy} by applying these properties to Equation~(\ref{subdifferentialequal}).

\end{proof}

\begin{lemma}\label{errorbound}
If $\mathbb{E}\left\{\left\|\sum_{i=1}^{N_s}\boldsymbol{Q}_{i,1}^\mathrm{H}\boldsymbol{W}\boldsymbol{X}^{\mathrm{H}}\boldsymbol{Q}_{i,2}^\mathrm{H}\right\|_{\mathcal{A}}^{*}\right\}\leq \epsilon$, the solution of Equation~(\ref{optprobnoisy}), {\it i.e.,} $\hat{\boldsymbol{H}}_v$ satisfies
\begin{equation}
    \mathbb{E}\left\{\left\Vert\left(g(\hat{\boldsymbol{H}}_v)-g({\boldsymbol{H}}_v)\right)\boldsymbol{X}\right\Vert^{2}_2\right\} \leq 2\epsilon\left\|\boldsymbol{H}_v\right\|_\mathcal{A}.\label{lemma2}
\end{equation}
\end{lemma}
\begin{proof}
Considering ${\boldsymbol{Y}}=g({\boldsymbol{H}}_v)\boldsymbol{X}+{\boldsymbol{W}}$, we have 
\begin{align}
\MoveEqLeft \left\|\left(g(\hat{\boldsymbol{H}}_v)-g({\boldsymbol{H}}_v)\right) \boldsymbol{X}\right\|_{F}^{2} &{} \nonumber\\
&=\left\langle g(\hat{\boldsymbol{H}}_v) \boldsymbol{X}-g({\boldsymbol{H}}_v) \boldsymbol{X}, g(\hat{\boldsymbol{H}}_v) \boldsymbol{X}-g({\boldsymbol{H}}_v) \boldsymbol{X}\right\rangle \nonumber\\
&=\left\langle g(\hat{\boldsymbol{H}}_v) \boldsymbol{X}-g({\boldsymbol{H}}_v) \boldsymbol{X}, g(\hat{\boldsymbol{H}}_v) \boldsymbol{X}-(\boldsymbol{Y}-\boldsymbol{W})\right\rangle \nonumber\\
&=\left\langle g(\hat{\boldsymbol{H}}_v) \boldsymbol{X}-g({\boldsymbol{H}}_v) \boldsymbol{X}, \boldsymbol{W}-(\boldsymbol{Y}-g(\hat{\boldsymbol{H}}_v) \boldsymbol{X})\right\rangle \nonumber\\
&=\left\langle g(\hat{\boldsymbol{H}}_v) \boldsymbol{X}-g({\boldsymbol{H}}_v) \boldsymbol{X}, \boldsymbol{W}-(\boldsymbol{Y}-g(\hat{\boldsymbol{H}}_v) \boldsymbol{X})\right\rangle_{R} \nonumber\\
&=\left\langle g({\boldsymbol{H}}_v) \boldsymbol{X},\boldsymbol{Y}-g(\hat{\boldsymbol{H}}_v)\boldsymbol{X}\right\rangle_{R} - \left\langle g({\boldsymbol{H}}_v) \boldsymbol{X},\boldsymbol{W}\right\rangle_{R} \nonumber\\
&\quad +\left\langle g(\hat{\boldsymbol{H}}_v) \boldsymbol{X},\boldsymbol{W}\right\rangle_{R}-\left\langle g(\hat{\boldsymbol{H}}_v) \boldsymbol{X},\boldsymbol{Y}-g(\hat{\boldsymbol{H}}_v)\boldsymbol{X}\right\rangle_{R} \nonumber\\
&\overset{\mathrm{(d)}}{\leq} 2\epsilon\left\|\boldsymbol{H}_v\right\|_\mathcal{A}.
\end{align}
if $\left\|\sum_{i=1}^{N_s}\boldsymbol{Q}_{i,1}^\mathrm{H}\boldsymbol{W}\boldsymbol{X}^{\mathrm{H}}\boldsymbol{Q}_{i,2}^\mathrm{H}\right\|_{\mathcal{A}}^{*}\leq \epsilon$. Note that, $\mathrm{(d)}$ holds with Equation~(\ref{lemma1}) and the H{\"o}lder inequality as optimality implies that
\begin{subequations}
\begin{align}
\left\langle g({\boldsymbol{H}}_v) \boldsymbol{X},\boldsymbol{Y}-g(\hat{\boldsymbol{H}}_v)\boldsymbol{X}\right\rangle_{R}\leq{}& \epsilon\left\|\boldsymbol{H}_v\right\|_\mathcal{A}, \label{eq:optimallityCond1}\\
- \left\langle g({\boldsymbol{H}}_v) \boldsymbol{X},\boldsymbol{W}\right\rangle_{R} \leq{}& \epsilon\left\|\boldsymbol{H}_v\right\|_\mathcal{A},\\
\left\langle g(\hat{\boldsymbol{H}}_v) \boldsymbol{X},\boldsymbol{W}\right\rangle_{R}\leq{}& \left\|\hat{\boldsymbol{H}}_v\right\|_\mathcal{A},\\
\left\langle\boldsymbol{Y}-g(\hat{\boldsymbol{H}}_v)\boldsymbol{X},g(\hat{\boldsymbol{H}}_v)\boldsymbol{X}\right\rangle_R={}& \epsilon\|\hat{\boldsymbol{H}}_v\|_{\mathcal{A}}.
\end{align}
\end{subequations}
We restrict the derivation of the previous to Equation~\eqref{eq:optimallityCond1} as the others can be shown analogously 
\begingroup
\allowdisplaybreaks
\begin{align}
\MoveEqLeft \left\langle g({\boldsymbol{H}}_v) \boldsymbol{X},\boldsymbol{Y}-g(\hat{\boldsymbol{H}}_v)\boldsymbol{X}\right\rangle_{R} & \nonumber\\
&{=}\left\langle {\boldsymbol{H}}_v ,\sum_{i=1}^{N_s}\boldsymbol{Q}_{i,1}^\mathrm{H}\left(\boldsymbol{Y}-g(\hat{\boldsymbol{H}}_v)\boldsymbol{X}\right)\boldsymbol{X}^{\mathrm{H}}\boldsymbol{Q}_{i,2}^\mathrm{H}\right\rangle_{R} \nonumber \\
&\overset{\mathrm{(e)}}{\leq}\left\|\boldsymbol{H}_v\right\|_\mathcal{A}\left\|\sum_{i=1}^{N_s}\boldsymbol{Q}_{i,1}^\mathrm{H}\left(\boldsymbol{Y}-g(\hat{\boldsymbol{H}}_v)\boldsymbol{X}\right)\boldsymbol{X}^{\mathrm{H}}\boldsymbol{Q}_{i,2}^\mathrm{H}\right\|_{\mathcal{A}}^{*} \nonumber\\
&\overset{\mathrm{(f)}}{\leq} \epsilon\left\|\boldsymbol{H}_v\right\|_\mathcal{A},
\end{align}
\endgroup
where $\mathrm{(e)}$ and $\mathrm{(f)}$ follow from H{\"o}lder inequality and Equation~(\ref{lemma1}). We conclude on Equation~(\ref{lemma2}) given $\mathbb{E}\left\{\left\|\sum_{i=1}^{N_s}\boldsymbol{Q}_{i,1}^\mathrm{H}\boldsymbol{W}\boldsymbol{X}^{\mathrm{H}}\boldsymbol{Q}_{i,2}^\mathrm{H}\right\|_{\mathcal{A}}^{*}\right\}\leq \epsilon$.
\end{proof}

\begin{proofp3} The goal of Proposition \ref{errorboundG} is to give an upper bound on the EER. Based on Lemma \ref{errorbound}, we need to derive a proper upper bound on $\mathbb{E}\left\{\left\|\sum_{i=1}^{N_s}\boldsymbol{Q}_{i,1}^\mathrm{H}\boldsymbol{W}\boldsymbol{X}^{\mathrm{H}}\boldsymbol{Q}_{i,2}^\mathrm{H}\right\|_{\mathcal{A}}^{*}\right\}$.
\begin{figure*}[!t]
\begin{equation}
\begin{aligned}
\left\|\sum_{i=1}^{N_s}\boldsymbol{Q}_{i,1}^\mathrm{H}\boldsymbol{W}\boldsymbol{X}^{\mathrm{H}}\boldsymbol{Q}_{i,2}^\mathrm{H}\right\|_\mathcal{A}^{*}&= \sup_{ \tau,\theta_\mathrm{Rx},\theta_\mathrm{Tx}} \left|\left\langle \tilde{\boldsymbol{W}}  , \frac{(N_s+1)\sqrt{N_rN_t}}{2}\boldsymbol{A}(\tau,\theta_\mathrm{Rx},\theta_\mathrm{Tx},0)\right\rangle\right|\\
&\overset{\mathrm{(g)}}{=}\sup_{\tilde{\bm{\delta}} \in \bm{\Gamma}} \left|\sum_{z=1}^{N_r}\sum_{t=1}^{N_t}\sum_{i=1}^{N_s}\sum_{o=1}^{\frac{N_s+1}{2}}\tilde{\boldsymbol{W}}[(o-1)N_r+z,(i-o)N_t+t]\tilde{\delta}_1^{i-1}\tilde{\delta}_2^{z-1}\tilde{\delta}_3^{t-1}\right|\\
&\overset{\mathrm{(h)}}{=}\sup_{\tilde{\bm{\delta}} \in \bm{\Gamma}} \left|\sum_{z=1}^{N_r}\sum_{t=1}^{N_t}\sum_{i=1}^{N_s}\tilde{\boldsymbol{W}}[(i-1)N_r+z,(i-1)N_t+t]\tilde{\delta}_1^{i-1}\tilde{\delta}_2^{z-1}\tilde{\delta}_3^{t-1}\right|=\sup_{\tilde{\bm{\delta}} \in \bm{\Gamma}}\left|\Xi(\tilde{\bm{\delta}})\right|
\end{aligned}\label{polynomial}
\end{equation}
\hrule
\end{figure*}
To that end, we start by introducing the quantity $\tilde{\bm{W}}$ as
\begin{equation}
    \tilde{\boldsymbol{W}}\triangleq\frac{2}{(N_s+1)\sqrt{N_rN_t}}\sum_{i=1}^{N_s}\boldsymbol{Q}_{i,1}^\mathrm{H}\boldsymbol{W}\boldsymbol{X}^{\mathrm{H}}\boldsymbol{Q}_{i,2}^\mathrm{H}.
\end{equation}
Next, let $\bm{\Gamma} = (0,1] \times (-\frac{1}{2}, \frac{1}{2}] \times (-\frac{1}{2}, \frac{1}{2}]$, $\bm{\delta} = [\delta_1, \delta_2, \delta_3] \in \bm{\Gamma}$ with $\delta_1 \triangleq {\frac{\tau}{NT_s}}$, $\delta_2 \triangleq {\frac{ d\sin(\theta_\mathrm{Rx})}{\lambda_c}}$, $\delta_3 \triangleq {\frac{ d\sin(\theta_\mathrm{Tx})}{\lambda_c}}$, and let $\tilde{\bm{\delta}} = [\tilde{\delta}_1, \tilde{\delta}_2, \tilde{\delta}_3] = [e^{i 2\pi \delta_1}, e^{i 2\pi \delta_2}, e^{i 2\pi \delta_3}]$. We define the function
\begin{equation}
\Xi(\tilde{\bm{\delta}}) 
\triangleq\sum_{\substack{z=1,\dots,N_r\\t=1,\dots,N_t\\i=1,\dots,N_s}} \tilde{\boldsymbol{W}}[(i-1)N_t+t,(i-1)N_r+z]\tilde{\delta}_1^{i-1}\tilde{\delta}_2^{z-1}\tilde{\delta}_3^{t-1}. 
\end{equation}


From Equation~(\ref{defdualnorm}), $\left\|\sum_{i=1}^{N_s}\boldsymbol{Q}_{i,1}^\mathrm{H}\boldsymbol{W}\boldsymbol{X}^{\mathrm{H}}\boldsymbol{Q}_{i,2}^\mathrm{H}\right\|_\mathcal{A}^{*}$ can be simplified as given in Equation~(\ref{polynomial}). Note that, $\mathrm{(g)}$ and $\mathrm{(h)}$ hold with Equation~(\ref{Qfunctions}) and Equation~(\ref{eq:atomicSet}). To clarify the proof, we introduce the quantity
$$\left\|\Xi\right\|_{\infty}\triangleq \sup_{\tilde{\bm{\delta}} \in \bm{\Gamma}}\left|\Xi(\tilde{\bm{\delta}})\right| = \left\|\sum_{i=1}^{N_s}\boldsymbol{Q}_{i,1}^\mathrm{H}\boldsymbol{W}\boldsymbol{X}^{\mathrm{H}}\boldsymbol{Q}_{i,2}^\mathrm{H}\right\|_\mathcal{A}^{*}.$$
Next, for any pair $\left(\bm{\delta}^{(1)},\bm{\delta}^{(2)}\right) \in \bm{\Gamma} \times \bm{\Gamma}$, we have that
\begin{align}
\MoveEqLeft\left|\Xi\left(\tilde{\bm{\delta}}^{(1)}\right)\right|-\left|\Xi\left( \tilde{\bm{\delta}}^{(2)} \right)\right|
\leq \left\|\boldsymbol{\tilde{\delta}}^{(1)}-\boldsymbol{\tilde{\delta}}^{(2)}\right\|_2 \sup_{\tilde{\boldsymbol{\delta}} \in \mathbb{T}^3}\left\|\frac{\partial\Xi}{\partial \tilde{\boldsymbol{\delta}}}\right\|_{2} \nonumber\\
&\overset{\mathrm{(i)}}{\leq} (N_s+N_r+N_t)\left\|\tilde{\boldsymbol{\delta}}^{(1)}-\tilde{\boldsymbol{\delta}}^{(2)}\right\|_2 \left\|\Xi\right\|_{\infty} \nonumber\\
&\leq (N_s+N_r+N_t) \left\|\tilde{\boldsymbol{\delta}}^{(1)}-\tilde{\boldsymbol{\delta}}^{(2)}\right\|_1 \left\|\Xi\right\|_{\infty} \nonumber\\
&\leq 2\pi(N_s+N_r+N_t) \left\|\boldsymbol{{\delta}}^{(1)}-\boldsymbol{{\delta}}^{(2)}\right\|_1 \left\|\Xi\right\|_{\infty},\label{diffineq}
\end{align}
where $\mathrm{(i)}$ holds according to the multi-dimensional Bernstein inequality~\cite{Tung}. To find the proper upper bound on $\left\|\Xi\right\|_{\infty}$, we extend~\cite[Appendix C]{bhaskar2013atomic} to our model.
Define the three-dimensional regular grid $\bm{G} \subset \bm{\Gamma}$ as
\begin{multline}
    \bm{G} = \frac{1}{N_{\mathrm{TOA}}} \Big\{0, \dots, N_{\mathrm{TOA}} \Big\} \\
    \times \frac{1}{N_{\mathrm{AOA}}} \left\{-\frac{N_{\mathrm{AOA}}}{2} + 1, \dots, \frac{N_{\mathrm{AOA}}}{2} \right\} \\
    \times \frac{1}{N_{\mathrm{AOD}}} \left\{-\frac{N_{\mathrm{AOD}}}{2} + 1, \dots, \frac{N_{\mathrm{AOD}}}{2} \right\}
\end{multline}
From Equation~(\ref{diffineq}), we can derive the inequality 
\begin{equation}
\begin{aligned}
    &\left\|\Xi\right\|_{\infty}\leq \max_{\tilde{\bm{\delta}}^{(2)}\in \bm{G}}\left|\Xi\left( \tilde{\bm{\delta}}^{(2)}\right)\right|\\
    &\quad+\pi(N_s+N_r+N_t) \left(\frac{1}{N_\mathrm{TOA}}+\frac{1}{N_\mathrm{AOA}}+\frac{1}{N_\mathrm{AOD}}\right) \left\|\Xi\right\|_{\infty}.
\end{aligned}
\end{equation}
Therefore, it comes that
\begin{equation}
\begin{aligned}
    &\left\|\Xi\right\|_{\infty}\\
    &\leq \left(1-\pi(N_s+N_r+N_t) \left(\frac{1}{N_\mathrm{TOA}}+\frac{1}{N_\mathrm{AOA}}+\frac{1}{N_\mathrm{AOD}}\right)\right)^{-1}\\
    &\qquad\times\max_{\tilde{\bm{\delta}}^{(2)} \in \bm{G}}\left|\Xi\left( \tilde{\bm{\delta}}^{(2)}\right)\right|\\
    &\overset{\mathrm{(j)}}{\leq} \left(1+2\pi(N_s+N_r+N_t) \left(\frac{1}{N_\mathrm{TOA}}+\frac{1}{N_\mathrm{AOA}}+\frac{1}{N_\mathrm{AOD}}\right)\right)\\
    &\qquad\times\max_{\tilde{\bm{\delta}}^{(2)} \in \bm{G}}\left|\Xi\left(\tilde{\bm{\delta}}^{(2)}\right)\right|,
\end{aligned}\label{discretean}
\end{equation}
where we assume that $\pi(N_s+N_r+N_t) \left(\frac{1}{N_\mathrm{TOA}}+\frac{1}{N_\mathrm{AOA}}+\frac{1}{N_\mathrm{AOD}}\right)< 1$. The inequality $\mathrm{(j)}$ can be verified with some algebra. 

Then, finding the upper bound on $\mathbb{E}\left\{\left\|\Xi\right\|_{\infty}\right\}$ amounts to evaluating the expectation of the right-hand side of Equation~(\ref{discretean}).
Recalling that $\boldsymbol{W}[i,j] \sim \mathcal{CN}({0},\sigma^2)$, we have from the definitions of $\boldsymbol{Q}_{i,1}$ and $\boldsymbol{Q}_{i,2}$ in Equation~(\ref{Qfunctions}) that
$
\Xi\left(\tilde{\bm{\delta}}^{(2)}\right)\sim\mathcal{CN}\left(0,\tilde\sigma^2\right)
$, with $
\tilde\sigma^2\triangleq {\frac{4\sigma^2\sum_{n=1}^{N_s}\sum_{g=1}^{G}\left(\sum_{t=1}^{N_t}\boldsymbol{X}[(n-1)N_t+t,g]\right)^2}{N_t(N_s+1)^2}}$.
By~\cite[Lemma 5]{bhaskar2013atomic} the expectation of $ \max_{\bm{\delta}^{(2)} \in \bm{G}}\left|\Xi\left(\tilde{\bm{\delta}}^{(2)}\right)\right|$ can be bounded by,
\begin{multline}
   \mathbb{E}\left\{\max_{\bm{\delta}^{(2)} \in \bm{G}}\left|\Xi\left( \tilde{\bm{\delta}}^{(2)}\right)\right|\right\} 
\leq \tilde{\sigma}\sqrt{\log(N_{\mathrm{TOA}}N_{\mathrm{AOA}}N_{\mathrm{AOD}})+ 1 }\label{upperboundonmaxdis}
\end{multline}
By substituting Equation~(\ref{upperboundonmaxdis}) into Equation~(\ref{discretean}), we can obtain the upper bound
\begin{align}
&\mathbb{E}\left\{\left\|\Xi\right\|_{\infty}\right\} \nonumber\\
\;&\leq \left(1+2\pi(N_s+N_r+N_t) \left(\frac{1}{N_\mathrm{TOA}}+\frac{1}{N_\mathrm{AOA}}+\frac{1}{N_\mathrm{AOD}}\right)\right)\nonumber \\
\;& \qquad \times\tilde{\sigma}\sqrt{\log(N_{\mathrm{TOA}}N_{\mathrm{AOA}}N_{\mathrm{AOD}})+ 1}.\label{generalupperbound}
\end{align} 
Furthermore, using the harmonic mean-geometric mean inequality, we have 
\begin{equation}
\begin{aligned}
&\mathbb{E}\left\{\left\|\Xi\right\|_{\infty}\right\}\leq\tilde{\sigma} \left(1+2\pi(N_s+N_r+N_t)\frac{1}{\tilde{N}}\right)\sqrt{\log(\tilde{N})+ 1},\\
\end{aligned}\label{generalupperbound2} 
\end{equation}
where $\tilde{N}\triangleq N_{\mathrm{TOA}}=N_{\mathrm{AOA}}=N_{\mathrm{AOD}}$.
Setting $\tilde N=2\pi(N_s+N_r+N_t)\log(N_s+N_r+N_t)$~\cite[Appendix D]{bhaskar2013atomic} in Equation~\eqref{generalupperbound2} yields the desired statement.
\end{proofp3}
\endgroup



\renewcommand*{\bibfont}{\footnotesize}
\printbibliography

@article{duval2020characterization,
  title={{A characterization of the Non-Degenerate Source Condition in super-resolution}},
  author={Duval, Vincent},
  journal={Information and Inference: A Journal of the IMA},
  volume={9},
  number={1},
  pages={235--269},
  year={2020},
  publisher={Oxford University Press}
}

@ARTICLE{tang_near_2015,
  author={Tang, Gongguo and Bhaskar, Badri Narayan and Recht, Benjamin},
  journal={IEEE Transactions on Information Theory}, 
  title={{Near Minimax Line Spectral Estimation}}, 
  year={2015},
  volume={61},
  number={1},
  pages={499-512},
  doi={10.1109/TIT.2014.2368122}}

@article{li_approximate_2018,
title = {{Approximate support recovery of atomic line spectral estimation: A tale of resolution and precision}},
journal = {Applied and Computational Harmonic Analysis},
volume = {48},
number = {3},
pages = {891-948},
year = {2020},
issn = {1063-5203},
doi = {https://doi.org/10.1016/j.acha.2018.09.005},
url = {https://www.sciencedirect.com/science/article/pii/S1063520318300824},
author = {Qiuwei Li and Gongguo Tang}
}

@INPROCEEDINGS{lemic2016localization,
  author={Lemic, Filip and Martin, James and Yarp, Christopher and Chan, Douglas and Handziski, Vlado and Brodersen, Robert and Fettweis, Gerhard and Wolisz, Adam and Wawrzynek, John},
  booktitle={2016 International Wireless Communications and Mobile Computing Conference (IWCMC)}, 
  title={{Localization as a feature of mmWave communication}}, 
  year={2016},
  volume={},
  number={},
  pages={1033-1038},
  doi={10.1109/IWCMC.2016.7577201}}

@article{chi2014compressive,
  title={{Compressive Two-Dimensional Harmonic Retrieval via Atomic Norm Minimization}},
  author={Chi, Yuejie and Chen, Yuxin},
  journal={IEEE Transactions on Signal Processing},
  volume={63},
  number={4},
  pages={1030--1042},
  year={2014},
  publisher={IEEE}
}

@article{lasserre2001global,
author = {Lasserre, Jean B.},
title = {{Global Optimization with Polynomials and the Problem of Moments}},
journal = {SIAM Journal on Optimization},
volume = {11},
number = {3},
pages = {796-817},
year = {2001},
doi = {10.1137/S1052623400366802},

URL = { 
        https://doi.org/10.1137/S1052623400366802
    
},
eprint = { 
        https://doi.org/10.1137/S1052623400366802
    
}
}

@inproceedings{da2018tight,
  author={Ferreira Da Costa, Maxime and Dai, Wei},
  booktitle={2018 IEEE International Symposium on Information Theory (ISIT)}, 
  title={{A Tight Converse to the Spectral Resolution Limit via Convex Programming}}, 
  year={2018},
  volume={},
  number={},
  pages={901-905},
  doi={10.1109/ISIT.2018.8437490}}

@article{da2020stable,
  author={Ferreira Da Costa, Maxime and Chi, Yuejie},
  journal={IEEE Transactions on Information Theory}, 
  title={{On the Stable Resolution Limit of Total Variation Regularization for Spike Deconvolution}}, 
  year={2020},
  volume={66},
  number={11},
  pages={7237-7252},
  doi={10.1109/TIT.2020.2993327}}

@article{bhaskar2013atomic,
  author={Bhaskar, Badri Narayan and Tang, Gongguo and Recht, Benjamin},
  journal={IEEE Transactions on Signal Processing}, 
  title={{Atomic Norm Denoising With Applications to Line Spectral Estimation}}, 
  year={2013},
  volume={61},
  number={23},
  pages={5987-5999},
  doi={10.1109/TSP.2013.2273443}}

@article{candes2014towards,
  title={{Towards a Mathematical Theory of Super-Resolution}},
  author={Cand{\`e}s, Emmanuel J and Fernandez-Granda, Carlos},
  journal={Communications on pure and applied Mathematics},
  volume={67},
  number={6},
  pages={906--956},
  year={2014},
  publisher={Wiley Online Library}
}

@article{chi2020harnessing,
  author={Chi, Yuejie and Ferreira Da Costa, Maxime},
  journal={IEEE Signal Processing Magazine}, 
  title={{Harnessing Sparsity Over the Continuum: Atomic norm minimization for superresolution}}, 
  year={2020},
  volume={37},
  number={2},
  pages={39-57},
  doi={10.1109/MSP.2019.2962209}}

@ARTICLE{Rappaport,
  author={Rappaport, Theodore S. and Sun, Shu and Mayzus, Rimma and Zhao, Hang and Azar, Yaniv and Wang, Kevin and Wong, George N. and Schulz, Jocelyn K. and Samimi, Mathew and Gutierrez, Felix},
  journal={IEEE Access}, 
  title={{Millimeter Wave Mobile Communications for 5G Cellular: It Will Work!}}, 
  year={2013},
  volume={1},
  number={},
  pages={335-349},
  doi={10.1109/ACCESS.2013.2260813}}

@ARTICLE{Hur,
  author={Hur, Sooyoung and Kim, Taejoon and Love, David J. and Krogmeier, James V. and Thomas, Timothy A. and Ghosh, Amitava},
  journal={IEEE Transactions on Communications}, 
  title={{Millimeter Wave Beamforming for Wireless Backhaul and Access in Small Cell Networks}}, 
  year={2013},
  volume={61},
  number={10},
  pages={4391-4403},
  doi={10.1109/TCOMM.2013.090513.120848}}

@INPROCEEDINGS{Hua,
  author={Hua Deng and Sayeed, Akbar},
  booktitle={2014 IEEE 15th International Workshop on Signal Processing Advances in Wireless Communications (SPAWC)}, 
  title={{mm-wave} {MIMO} channel modeling and user localization using sparse beamspace signatures}, 
  year={2014},
  volume={},
  number={},
  pages={130-134},
  doi={10.1109/SPAWC.2014.6941331}}

@INPROCEEDINGS{Saloranta,
  author={Saloranta, Jani and Destino, Giuseppe},
  booktitle={2016 24th European Signal Processing Conference (EUSIPCO)}, 
  title={{On the utilization of MIMO-OFDM channel sparsity for accurate positioning}}, 
  year={2016},
  volume={},
  number={},
  pages={748-752},
  doi={10.1109/EUSIPCO.2016.7760348}}

@ARTICLE{Shahmansoori,
  author={Shahmansoori, Arash and Garcia, Gabriel E. and Destino, Giuseppe and Seco-Granados, Gonzalo and Wymeersch, Henk},
  journal={IEEE Transactions on Wireless Communications}, 
  title={{Position and Orientation Estimation Through Millimeter-Wave MIMO in 5G Systems}}, 
  year={2018},
  volume={17},
  number={3},
  pages={1822-1835},
  doi={10.1109/TWC.2017.2785788}}

@INPROCEEDINGS{Duarte,
  author={Duarte, M.F. and Sarvotham, S. and Baron, D. and Wakin, M.B. and Baraniuk, R.G.},
  booktitle={Conference Record of the Thirty-Ninth Asilomar Conference onSignals, Systems and Computers, 2005.}, 
  title={{Distributed Compressed Sensing of Jointly Sparse Signals}}, 
  year={2005},
  volume={},
  number={},
  pages={1537-1541},
  doi={10.1109/ACSSC.2005.1600024}}

@ARTICLE{TangG,
  author={Tang, Gongguo and Bhaskar, Badri Narayan and Shah, Parikshit and Recht, Benjamin},
  journal={IEEE Transactions on Information Theory}, 
  title={{Compressed Sensing Off the Grid}}, 
  year={2013},
  volume={59},
  number={11},
  pages={7465-7490},
  doi={10.1109/TIT.2013.2277451}}

@ARTICLE{Li,
  author={Li, Jianxiu and Mitra, Urbashi},
  journal={IEEE Transactions on Communications}, 
  title={{Improved Atomic Norm Based Time-Varying Multipath Channel Estimation}}, 
  year={2021},
  volume={69},
  number={9},
  pages={6225-6235},
  doi={10.1109/TCOMM.2021.3087793}}

@ARTICLE{Wu1,
  author={Wu, Xiaohuan and Zhu, Wei-Ping},
  journal={IEEE Transactions on Vehicular Technology}, 
  title={{Single Far-Field or Near-Field Source Localization With Sparse or Uniform Cross Array}}, 
  year={2020},
  volume={69},
  number={8},
  pages={9135-9139},
  doi={10.1109/TVT.2020.2998128}}

@INPROCEEDINGS{Wu2,
  author={Wu, Xiaohuan and Zhu, Wei-Ping and Yan, Jun},
  booktitle={ICASSP 2020 - 2020 IEEE International Conference on Acoustics, Speech and Signal Processing (ICASSP)}, 
  title={{Atomic Norm Based Localization of Far-Field and Near-Field Signals with Generalized Symmetric Arrays}}, 
  year={2020},
  volume={},
  number={},
  pages={4762-4766},
  doi={10.1109/ICASSP40776.2020.9052922}}

@ARTICLE{TangW,
  author={Tang, Wen-Gen and Jiang, Hong and Zhang, Qi},
  journal={IEEE Signal Processing Letters}, 
  title={{Range-Angle Decoupling and Estimation for FDA-MIMO Radar via Atomic Norm Minimization and Accelerated Proximal Gradient}}, 
  year={2020},
  volume={27},
  number={},
  pages={366-370},
  doi={10.1109/LSP.2020.2972470}}

@ARTICLE{Tsai,
  author={Tsai, Yingming and Zheng, Le and Wang, Xiaodong},
  journal={IEEE Transactions on Communications}, 
  title={{Millimeter-Wave Beamformed Full-Dimensional MIMO Channel Estimation Based on Atomic Norm Minimization}}, 
  year={2018},
  volume={66},
  number={12},
  pages={6150-6163},
  doi={10.1109/TCOMM.2018.2864737}}

@article{Stoica,
title = {{On reparametrization of loss functions used in estimation and the invariance principle}},
journal = {Signal Processing},
volume = {17},
number = {4},
pages = {383-387},
year = {1989},
issn = {0165-1684},
doi = {https://doi.org/10.1016/0165-1684(89)90123-0},
author = {Petre Stoica and Torsten Söderström}
}

@ARTICLE{Yang1,
  author={Yang, Zai and Xie, Lihua},
  journal={IEEE Transactions on Signal Processing}, 
  title={{Exact Joint Sparse Frequency Recovery via Optimization Methods}}, 
  year={2016},
  volume={64},
  number={19},
  pages={5145-5157},
  doi={10.1109/TSP.2016.2576422}}

@INPROCEEDINGS{Cho,
  author={Cho, Myung and Cai, Jian-Feng and Liu, Suhui and Eldar, Yonina C. and Xu, Weiyu},
  booktitle={2016 IEEE International Conference on Acoustics, Speech and Signal Processing (ICASSP)}, 
  title={{Fast alternating projected gradient descent algorithms for recovering spectrally sparse signals}}, 
  year={2016},
  volume={},
  number={},
  pages={4638-4642},
  doi={10.1109/ICASSP.2016.7472556}}

@ARTICLE{Yang2,
  author={Yang, Zai and Xie, Lihua and Stoica, Petre},
  journal={IEEE Transactions on Information Theory}, 
  title={{Vandermonde Decomposition of Multilevel Toeplitz Matrices With Application to Multidimensional Super-Resolution}}, 
  year={2016},
  volume={62},
  number={6},
  pages={3685-3701},
  doi={10.1109/TIT.2016.2553041}}

@article{Fletcher,
  title={{A Modifidied Marquart Subroutine for Nonlinear Least Squares}},
  author={R. Fletcher},
  journal={Harwell Report},
  year={1971}
}

@book{Goldsmith, 
	place={Cambridge}, 
	title={{Wireless Communications}}, 
	DOI={10.1017/CBO9780511841224}, 
	publisher={Cambridge University Press}, 
	author={Goldsmith, Andrea}, 
	year={2005}}

@ARTICLE{Beygi,
  author={Beygi, Sajjad and Elnakeeb, Amr and Choudhary, Sunav and Mitra, Urbashi},
  journal={IEEE Transactions on Signal Processing}, 
  title={{Bilinear Matrix Factorization Methods for Time-Varying Narrowband Channel Estimation: Exploiting Sparsity and Rank}}, 
  year={2018},
  volume={66},
  number={22},
  pages={6062-6075},
  doi={10.1109/TSP.2018.2872886}}

@ARTICLE{Elnakeeb,
  author={Elnakeeb, Amr and Mitra, Urbashi},
  journal={IEEE Transactions on Signal Processing}, 
  title={{Bilinear Channel Estimation for {MIMO OFDM}: Lower Bounds and Training Sequence Optimization}}, 
  year={2021},
  volume={69},
  number={},
  pages={1317-1331},
  doi={10.1109/TSP.2021.3056591}}

@INBOOK{Zekavat,
  author={Zekavat, Reza and Buehrer, R. Michael},
  title={{Handbook of Position Location: Theory, Practice, and Advances}}, 
  publisher={Wiley-IEEE Press},
  year={2019},
  volume={},
  number={},
  pages={435-465},
  doi={10.1002/9781119434610.ch13}}

@INPROCEEDINGS{Heckel,
  author={Heckel, Reinhard},
  booktitle={2016 IEEE International Symposium on Information Theory (ISIT)}, 
  title={{Super-resolution MIMO radar}}, 
  year={2016},
  volume={},
  number={},
  pages={1416-1420},
  doi={10.1109/ISIT.2016.7541532}}

@Article{Zhao,
AUTHOR = {Zhao, Kun and Zhao, Tiantian and Zheng, Zhengqi and Yu, Chao and Ma, Difeng and Rabie, Khaled and Kharel, Rupak},
TITLE = {{Optimization of Time Synchronization and Algorithms with TDOA Based Indoor Positioning Technique for Internet of Things}},
JOURNAL = {Sensors},
VOLUME = {20},
YEAR = {2020},
NUMBER = {22},
ARTICLE-NUMBER = {6513},
ISSN = {1424-8220},
DOI = {10.3390/s20226513}
}

@article{Tung,
 ISSN = {00029939, 10886826},
 author = {S. H. Tung},
 journal = {Proceedings of the American Mathematical Society},
 number = {1},
 pages = {73--76},
 publisher = {American Mathematical Society},
 title = {{Bernstein's Theorem for the Polydisc}},
 volume = {85},
 year = {1982}
}

@INPROCEEDINGS{Liicassp,
  author={Li, Jianxiu and Ferreira Da Costa, Maxime and Mitra, Urbashi },
  booktitle={ICASSP 2022 - 2022 IEEE International Conference on Acoustics, Speech and Signal Processing (ICASSP)}, 
  title={{Atomic Norm Based Localization and Orientation Estimation for Millimeter-Wave MIMO OFDM Systems}}, 
  year={2022},
  volume={},
  number={},
  }

@INPROCEEDINGS{Zhou,
  author={Zhou, Bingpeng and Wichman, Risto and Zhang, Lei and Luo, Zhiyong},
  booktitle={2021 IEEE 32nd Annual International Symposium on Personal, Indoor and Mobile Radio Communications (PIMRC)}, 
  title={{Simultaneous Localization and Channel Estimation for 5G mmWave MIMO Communications}}, 
  year={2021},
  volume={},
  number={},
  pages={1208-1214},
  doi={10.1109/PIMRC50174.2021.9569680}}

@ARTICLE{Fan,
  author={Fan, Jiancun and Dou, Xiaoyuan and Zou, Wei and Chen, Shijun},
  journal={IEEE Transactions on Vehicular Technology}, 
  title={{Localization Based on Improved Sparse Bayesian Learning in mmWave MIMO Systems}}, 
  year={2022},
  volume={71},
  number={1},
  pages={354-361},
  doi={10.1109/TVT.2021.3123147}}

@INPROCEEDINGS{Zheng,
  author={Zheng, Wenqing and González-Prelcic, Nuria},
  booktitle={2019 53rd Asilomar Conference on Signals, Systems, and Computers}, 
  title={{Joint Position, Orientation AND Channel Estimation in Hybrid mmWAVE MIMO Systems}}, 
  year={2019},
  volume={},
  number={},
  pages={1453-1458},
  doi={10.1109/IEEECONF44664.2019.9048769}}

@INPROCEEDINGS{Zhu,
  author={Zhu, Feibai and Liu, An and Lau, Vincent K. N.},
  booktitle={ICC 2019 - 2019 IEEE International Conference on Communications (ICC)}, 
  title={{Channel Estimation and Localization for mmWave Systems: A Sparse Bayesian Learning Approach}}, 
  year={2019},
  volume={},
  number={},
  pages={1-6},
  doi={10.1109/ICC.2019.8761825}}

\end{document}